\pdfoutput=1
\documentclass[letterpaper,11pt]{article}
\RequirePackage{etex}
\usepackage[utf8]{inputenc}
\usepackage[whole]{bxcjkjatype}
\usepackage[in]{fullpage}
\usepackage{todonotes,comment}
\usepackage{amsmath,amsthm,amssymb,amsfonts,mathtools,bm}
\usepackage[linesnumbered,ruled,vlined]{algorithm2e}
\usepackage[colorlinks=true,linkcolor=blue,citecolor=blue,hypertexnames=false]{hyperref}
\usepackage{cleveref}
\usepackage{autonum}
\newcommand{\BIB}[1]{}
\usepackage{authblk}
\usepackage{xcolor}
\definecolor{gold}{RGB}{236,225,48}

\newtheorem{theorem}{Theorem}[section]
\newtheorem{lemma}[theorem]{Lemma}
\newtheorem{proposition}[theorem]{Proposition}
\newtheorem{claim}[theorem]{Claim}
\newtheorem{example}[theorem]{Example}
\newtheorem{corollary}[theorem]{Corollary}

\usepackage{tikz}
\usepackage{booktabs}
\definecolor{myyellow}{cmyk}{0,0.02,0.23,0.01}

\DeclareMathOperator{\supp}{supp}
\DeclareMathOperator*{\argmax}{arg\,max}
\DeclareMathOperator*{\argmin}{arg\,min}

\newcommand{\B}{\mathbf{B}}

\newcommand{\cL}{\mathcal{L}}
\newcommand{\rL}{\mathrm{L}}

\renewcommand{\epsilon}{\varepsilon}
\newcommand{\ot}{\leftarrow}
\newcommand{\conv}{\mathop{\mathrm{conv}}}
\renewcommand{\mid}{:}
\renewcommand{\bar}[1]{\mkern 1.5mu\overline{\mkern-1.5mu#1\mkern-1.5mu}\mkern 1.5mu}
\renewcommand{\underbar}[1]{\mkern 1.5mu\underline{\mkern-1.5mu#1\mkern-1.5mu}\mkern 1.5mu}

\title{Fair Allocation with Binary Valuations \mbox{for Mixed Divisible and Indivisible Goods}}
\author[1]{Yasushi Kawase}
\affil{University of Tokyo}
\author[2]{Koichi Nishimura}
\affil{CRESCO LTD.}
\author[3]{Hanna Sumita}
\affil{Tokyo Institute of Technology}
\date{}
\begin{document}
\maketitle

\begin{abstract}
The fair allocation of mixed goods, consisting of both divisible and indivisible goods, has been a prominent topic of study in economics and computer science. We define an allocation as fair if its utility vector minimizes a symmetric strictly convex function. This fairness criterion includes standard ones such as maximum egalitarian social welfare and maximum Nash social welfare. We address the problem of minimizing a given symmetric strictly convex function when agents have binary valuations. If only divisible goods or only indivisible goods exist, the problem is known to be solvable in polynomial time. In this paper, firstly, we demonstrate that the problem is NP-hard even when all indivisible goods are identical. This NP-hardness is established even for maximizing egalitarian social welfare or Nash social welfare. Secondly, we provide a polynomial-time algorithm for the problem when all divisible goods are identical. To accomplish these, we exploit the proximity structure inherent in the problem. This provides theoretically important insights into the hybrid domain of convex optimization that incorporates both discrete and continuous aspects.
\end{abstract}

\section{Introduction}
The fair allocation of resources is a fundamental problem that arises in various real-life scenarios, such as dividing tasks among team members, distributing inheritance, or allocating land.
In many of these situations, the resources can be of mixed types: divisible and indivisible. 
The investigation of fair allocation mechanisms for mixed goods has thus gained considerable interest recently~\cite{Bei2021,Bhaskar+2021,Caragiannis2019,LiLLT2023,Lu+2023,NS2023,liu2023mixed}.

In this paper, we study the following fair allocation problem of mixed goods.
We are given a set $N=\{1,2,\dots,n\}$ of $n$ agents and a set $E=C\cup M$, where $C$ and $M$ are the sets of divisible and indivisible goods, respectively. 
We are also given a \emph{symmetric strictly convex} function $\Phi\colon \mathbb{R}^N \to \mathbb{R}$.
Each agent $i$ has a binary valuation $v_{ie} \in \{0,1\}$ for each good $e$, meaning that for each good, each agent either desires it or not.
An allocation is a matrix $\pi \in [0,1]^{N \times E}$ such that $\pi_{ie}\in\{0,1\}$ for all $i\in N$ and $e\in M$.
Throughout this paper, we only consider utilitarian optimum allocations, that is, $\pi_{ie}>0$ only if $v_{ie}=1$.
The entry $\pi_{ie}$ means the allocated amount of good $e$ to agent $i$.
Agents have additive utility, and the utility of agent $i$ in allocation $\pi$ is $\pi_i(E)=\sum_{e\in E} v_{ie}\pi_{ie}$. 
For an allocation $\pi$, a vector $z=(\pi_1(E),\dots,\pi_n(E))$ is called a utility vector of $\pi$.
We say that an allocation $\pi$ is \emph{$\Phi$-fair} if its utility vector $z$ minimizes $\Phi(z)$ among allocations.
The goal of our problem is to find a $\Phi$-fair allocation.

This problem is a generalization of maximizing fairness measures such as the \emph{Nash welfare} and the \emph{egalitarian social welfare} (max-min fairness), as we will see in Section~\ref{sec:preliminaries}.
Roughly speaking, the former notion is defined as the product of positive utilities $\prod_{i\in N\colon z_i>0} z_i$, and the latter is the minimum utility among agents $\min_{i\in N} z_i$.

There is a vast body of literature on the allocation of goods in cases where only divisible or indivisible goods are present.
In such cases, it suffices to find a feasible minimizer (utility vector) of $\Phi$ because we can find an allocation achieving the utilities by solving the maximum flow problem (see also Section~\ref{sec:preliminaries}).
In the continuous case, where there are only divisible goods, the set of possible utility vectors forms an \emph{integral base-polyhedron}\footnote{An integral base-polyhedron is a polyhedron that can be represented as $\{x\in\mathbb{R}^N\mid x(N)=f(N)\text{ and }x(X)\le f(X)~(\forall X\subseteq N)\}$ by an integer-valued submodular function $f\colon 2^N\to\mathbb{Z}$ with $f(\emptyset)=0$. See \Cref{subsec:definitions} for details.}.
It is known that an integral base-polyhedron has a common unique minimizer independent of $\Phi$, and the minimizer can be characterized by a structure called the \emph{principal partition}~\cite{fujishige1980,Maruyama1978} (see \Cref{subsec:partition} for the definition and details).
Since there exist polynomial-time algorithms to find an allocation achieving the maximum Nash welfare (MNW)~\cite{Orlin2010,Vegh2016}, a minimizer of $\Phi$ can be found in polynomial time.
In the discrete case, where there are only indivisible goods, the set of possible utility vectors forms an \emph{M-convex set}, which is the set of integral vectors in an integral base-polyhedron. 
It is known that a minimizer of $\Phi$ on an M-convex set can be characterized by the \emph{canonical partition}~\cite{FM2022a}, which is an aggregation of the principal partition.
Additionally, the set of minimizers of a symmetric strictly convex function does not depend on the function~\cite{FM2022a} and a minimizer (utility vector) can be found in polynomial time~\cite{FM2022b}.
Furthermore, a \emph{proximity theorem} has been established~\cite{FMdecmin2}. 
This theorem states that a minimizer of $\Phi$ in an M-convex set lies within a unit hypercube that contains the minimizer in the corresponding integral base-polyhedron.

Our problem is regarded as the hybrid of continuous and discrete optimization problems; finding an allocation whose utility vector $z$ minimizes $\Phi(z)$ under the constraint that $z$ belongs to the Minkowski sum of an integral base-polyhedron and an M-convex set.
Unfortunately, the hybrid case may not inherit nice properties of continuous or discrete cases.
When there are both divisible and indivisible goods, the set of possible utility vectors is not necessarily an integral base-polyhedron or an M-convex set. 
It also does not work to find allocations of divisible and indivisible goods separately and combine them.
We can observe these from the following example.
\begin{example}\label{ex:neither}
Suppose that there are one indivisible good $g$, one divisible good $c$, and three agents who desire both goods.
Let $\Phi(z)= - z_1 \cdot z_2 \cdot z_3$.
In this case, allocating $c$ equally to the three agents minimizes $\Phi$ when considering only $c$.
However, allocating $g$ to agent $1$ and $c$ to agents $2$ and $3$ equally minimizes $\Phi$ for mixed goods.
In addition, the set of possible utility vectors is not an M-convex set since it contains fractional utility vectors and not an integral base-polyhedron since it is not convex (see \Cref{fig:neither}). 
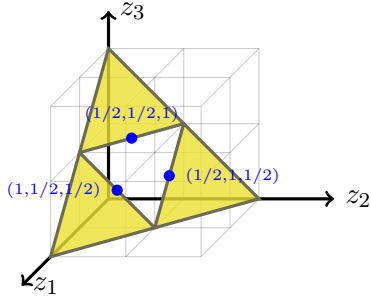
\begin{figure}[htbp]
\centering
\begin{tikzpicture}
		[cube/.style={thin,black,opacity=.2},
			util/.style={very thick,yellow!30!black,fill=gold,fill opacity=.8},
			axis/.style={->,black,very thick},
            p/.style={circle,fill=blue,inner sep=1.5pt}]
	\draw[axis] (0,0,0) -- (0,0,3) node[anchor=west]{$z_1$};
	\draw[axis] (0,0,0) -- (3,0,0) node[anchor=west]{$z_2$};
	\draw[axis] (0,0,0) -- (0,2.5,0) node[anchor=west]{$z_3$};
    \foreach \x in {0,1,2}{
        \foreach \y in {0,1,2}{
            \draw[cube] (\x,\y,0) -- (\x,\y,2);
            \draw[cube] (\x,0,\y) -- (\x,2,\y);
            \draw[cube] (0,\x,\y) -- (2,\x,\y);
        }
    }
    \draw[util] (2,0,0) -- (1,1,0) -- (1,0,1) -- cycle;
    \draw[util] (0,2,0) -- (1,1,0) -- (0,1,1) -- cycle;
    \draw[util] (0,0,2) -- (1,0,1) -- (0,1,1) -- cycle;
    \node[p,label={[blue]right:\tiny (1/2,1,1/2)}] at (1,1/2,1/2) {};
    \node[p,label={[blue]left:\tiny (1,1/2,1/2)}] at (1/2,1/2,1) {};
    \node[p,label={[blue]above:\tiny (1/2,1/2,1)}] at (1/2,1,1/2) {};
\end{tikzpicture}
\caption{The set of possible utility vectors in Example~\ref{ex:neither}. The blue points are minimizers of $\Phi$.}\label{fig:neither}
\end{figure}
\end{example}
Therefore, existing results are not applicable to our problem.
In addition, even if we obtain a minimizer of $\Phi$, it is not clear how to derive an allocation that achieves the utility vector.
Thus, investigating a structure and establishing a polynomial-time algorithm for the hybrid problem is much more challenging.

\subsection{Our contribution}
First, we investigate the structure of the fair allocation of mixed goods.
Unfortunately, certain nice properties that hold in the divisible case or the indivisible case no longer hold in the mixed case, as we will see in \Cref{ex:decmin=SOS,ex:incmax=SOS}.
Nevertheless, we demonstrate that a proximity theorem (\Cref{thm:proximity}) holds in the mixed case.
The proximity theorem states that a minimizer of $\Phi$ for mixed goods lies within a unit box containing the minimizer of $\Phi$ assuming that every good is divisible.
This generalizes a proximity theorem for the discrete case~\cite{FMdecmin2}.
It is worth mentioning that our theorem holds even when each agent evaluates indivisible goods with a matroid rank function and divisible goods with the concave closure of a matroid rank function.
By the proximity theorem, an optimal integral solution is a good approximation solution for mixed goods, in the sense that the $\ell_\infty$ distance from the optimal solution is at most 1.

Note that our proof does not depend on the fair allocation setting.
Therefore, these properties are fundamental even for the hybrid of continuous and discrete optimization.
Our proximity theorem cannot be proven by directly applying the proof to the discrete-only or the continuous-only case since the structure is different. 
To prove it, we introduce new exchange properties. 
The main idea is to see an optimality criterion in terms of an exchange graph.
Unlike the discrete case, the situation of the mixed case is much complicated.
Our main contribution to this paper is to provide an elaborate analysis of the graph.
As a consequence, we can show that the hybrid case still retains a structure of the canonical partition (\Cref{lem:decomposition}). 

Next, by utilizing the proximity theorem, we show that our problem is NP-hard even when indivisible goods are \emph{identical}, i.e., for each agent $i$, either $v_{ie}=1$ ($\forall e\in M$) or $v_{ie}=0$ ($\forall e \in M$).
\begin{theorem}\label{thm:intro-NP-hard}
For any fixed symmetric strictly convex function $\Phi$, finding a $\Phi$-fair allocation is NP-hard even when indivisible goods are identical.
\end{theorem}
We also prove that computing an MNW allocation and an optimal egalitarian allocation are both NP-hard. These results highlight the difficulty of the mixed goods case because the problems can be solved in polynomial time when there are only divisible goods or only indivisible goods.

Finally, we show the following tractability when divisible goods are identical.
\begin{theorem}\label{thm:intro-alg}
Let $\Phi$ be a symmetric strictly convex function.
There exists a polynomial-time algorithm that finds a $\Phi$-fair allocation if all the divisible goods are identical.
\end{theorem}
A key tool to construct our algorithm is the canonical partition for the mixed goods.
By applying it, we can partition goods as $E_1, \dots, E_q$ and agents as $N_1,\dots, N_q$ so that goods in $E_j$ are allocated to agents in $N_j$ in a $\Phi$-fair allocation (\Cref{thm:structure fair}).
Thanks to this structure, a minimizer of $\Phi$ can be found by independently solving the subproblems of assigning $E_j$ to $N_j$ for $j=1,2,\dots,q$.
In each subproblem, the utility of every agent is almost the same.
However, unlike the continuous or discrete case, an optimal allocation depends on $\Phi$~(see \Cref{ex:decmin=SOS,ex:incmax=SOS}). 
Thus, it is not easy to obtain a full characterization of minimizers.

\subsection{Related work}
For the fair allocation of divisible homogeneous goods with additive valuations (not restricted to binary), an MNW allocation corresponds to a market equilibrium of a special case of the Fisher's market model (see, e.g.,~\cite{textbook}). 
Moreover, an MNW allocation is \emph{envy-free} (EF)~\cite{Segal-Halevi2019,Varian}, that is, no agent envies any other agent.
It is known that this problem can be solved in strongly polynomial-time~\cite{Orlin2010,Vegh2016}.

For the fair allocation of indivisible goods with additive valuations, Caragiannis et al.~\cite{Caragiannis2019} proved that an MNW allocation is \emph{envy-free up to one good} (EF1), that is, each agent $i$ does not envy another agent $j$ if some indivisible good is removed from the bundle of agent $j$.
Since computing an MNW allocation is hard in general~\cite{Lee2017}, there is a series of research to design an approximation algorithm~\cite{Anari2018,Cole2017,Cole2015,Cole2018,Garg2018}.
Benabbou et al.~\cite{Benabbou2021} proved that the set of MNW allocations coincides with that of minimizers of any symmetric strictly convex function, even when the utility of each agent is represented by a matroid rank function.\footnote{Note that \Cref{thm:proximity} is an extension of this result. Specifically, we construct an ``augmenting'' path of \cite[Section 3.2]{Benabbou2021} for a hybrid situation.}
Harvey et al.~\cite{HLLT2006} proposed efficient algorithms for computing an allocation that minimizes a certain symmetric strictly convex function.
When agents' valuations are matroid rank function, Goko et al.~\cite{Goko+2022} presented a truthful and EF mechanism by subsidizing each agent with at most $1$. 
Note that subsidies can be viewed as divisible goods that every agent desires.
When agents have binary additive valuations, an MNW allocation can be computed in polynomial time~\cite{Barman2018,Darman2015}.
Truthful mechanisms to find an MNW allocation are also proposed~\cite{halpern2020fair,babaioff2020fair}.

Fair allocation with a mixture of divisible and indivisible goods has recently gained attention and has been the subject of research.
Bei et al.~\cite{Bei2021} introduced a fairness notion called \emph{envy-freeness for mixed goods} (EFM) as a generalization of EF and EF1 notions. 
Caragiannis et al.~\cite{Caragiannis2019} mentioned that an MNW allocation is \emph{envy-free up to one good for mixed goods} (EF1M), which is a relaxation of EFM.
Very recently, Li et al.~\cite{LiLLT2023} proposed a truthful mechanism that outputs an EFM allocation for the case where agents have binary additive valuations on indivisible goods and a common valuation on a single divisible good (e.g., money).
They also showed that their mechanism runs in polynomial time, and its output achieves MNW.
We remark that their algorithm does not work in our problem even when divisible goods are identical because we allow some agents to have value $0$ on them. 
For more details, see a survey paper by Liu et al.~\cite{liu2023mixed}.

The (integral) base-polyhedron has been studied in the theory of matroids and submodular functions~\cite{fujishige2005}.
The concept of M-convex sets was introduced by Murota~\cite{Murota1998} and is defined as a set of integral vectors satisfying certain exchange axioms.
\emph{Discrete convex analysis}~\cite{Murota:DCA} is a framework of convex analysis in discrete settings, including M-convexity.

The concepts of continuous/discrete hybrid convexity have been proposed by Takamatsu et al.~\cite{THM2004} and Moriguchi et al.~\cite{MHM2007}.
In particular, Moriguchi et al.~\cite{MHM2007} provided an optimality criterion for an integral polyhedral hybrid M-convex function minimization.
However, the functions treated in the present paper are hybrid M-convex functions that are not necessarily integral polyhedral.

\subsection{Organization}
The rest of this paper is organized as follows. 
Firstly, the formal definition of our problem is provided in \Cref{sec:preliminaries}. 
Next, in \Cref{sec:M-convex}, we introduce integral base-polyhedra and M-convex sets. We formulate our fair allocation problem with these concepts.
Then, in \Cref{sec:structure}, we present a proximity structure that aids in understanding the problem. 
\Cref{sec:identical-divisible} provides a polynomial-time algorithm for the case where all divisible goods are identical.
\Cref{sec:identical-indivisible} proves NP-hardness even when all indivisible goods are identical.
We defer some proofs to Appendix.

\section{Preliminaries}\label{sec:preliminaries}
For $k\in\mathbb{N}$, we denote $[k]=\{1,2,\dots,k\}$.
Let $N=[n]$ represent the set of $n$ agents. We have two types of goods: $M=\{g_1,g_2,\dots,g_m\}$ represents the set of indivisible goods, and $C=\{c_1,c_2,\dots,c_r\}$ denotes the set of homogeneous divisible goods, that is, the valuation for a piece of a good is proportional to its fraction.
The set of all goods is denoted by $E = M \cup C$. 
Let $v_{ie}$ be the valuation of good $e\in E$ for agent $i\in N$.
Throughout this paper, we assume that agents have binary valuations, that is, the valuation $v_{ie}$ for the whole of good $e$ is either $0$ or $1$ for all $i\in N$ and $e\in E$.
An instance of the fair allocation we deal with in this paper is described as $(N,M,C,v)$.
Without loss of generality, we assume that, for any $e\in E$, there exists $i\in N$ such that $v_{ie}=1$.

A \emph{relaxed allocation} is defined as a matrix $\pi\in[0,1]^{N\times E}$ that satisfies
(i) $\sum_{i\in N}\pi_{ie}=1$ for all $e\in E$ and
(ii) $\pi_{ie}=0$ for any $i\in N$ and $e\in E$ with $v_{ie}=0$.
In a relaxed allocation $\pi$, each agent $i$ receives each good $e$ in the proportion of $\pi_{ie}$. Relaxed allocations treat indivisible goods as divisible.
A relaxed allocation $\pi$ is an \emph{allocation} if it additionally satisfies $\pi_{ie}\in\{0,1\}$ for all $i\in N$ and $e\in M$.
A relaxed allocation $\pi$ is an \emph{integral allocation} if it additionally satisfies $\pi_{ie}\in\{0,1\}$ for all $i\in N$ and $e\in E$.
For an allocation $\pi$, an agent $i\in N$, and a subset of goods $E'\subseteq E$, 
let $\pi_i(E')=\sum_{e\in E'}\pi_{ie}$, which is the valuation of agent $i$'s bundle from $E'$.
For an allocation $\pi$, the utility of agent $i\in N$ is defined as $\pi_i(E)$.
Our goal is to find a fair allocation.
As fairness measures, we employ $\Phi$-fairness defined below.

For an allocation $\pi$, let $\pi(E)$ be the utility vector $(\pi_1(E),\dots,\pi_n(E))$. 
For a utility vector $x$, let $x^\downarrow$ be the vector obtained from $x$ by rearranging its components in the decreasing order.
We call two vectors $x,y\in\mathbb{R}^N$ \emph{value-equivalent} if $x^\downarrow=y^\downarrow$.
A vector $x\in\mathbb{R}^N$ is \emph{decreasingly smaller} than a vector $y\in\mathbb{R}^N$ if $x^\downarrow$ is lexicographically smaller than $y^\downarrow$ (i.e., $x^\downarrow_1 < y^\downarrow_1$, or 
$x^\downarrow_1 = y^\downarrow_1$ and
$x^\downarrow_2 < y^\downarrow_2$, etc).
An allocation $\pi$ is called \emph{decreasingly minimal} (\emph{dec-min}, for short) if $\pi(E)$ is decreasingly smaller than or value-equivalent to $\pi'(E)$ for every allocation $\pi'$.
In other words, an allocation $\pi$ is dec-min if its largest utility is as small as possible, within this, its second largest utility (with the same or smaller value than the largest one) is as small as possible, and so on.
Similarly, an allocation $\pi$ is called \emph{increasingly maximal} (\emph{inc-max}, for short) if its smallest utility is as large as possible, within this, its second smallest utility is as large as possible, and so on.

We say that a function $\Phi\colon \mathbb{R}^N\to\mathbb{R}$ is \emph{symmetric} if 
\begin{align}
\Phi(z_1,z_2,\dots,z_n)=\Phi(z_{\sigma(1)},z_{\sigma(2)},\dots,z_{\sigma(n)})
\end{align}
for all permutations $\sigma$ of $(1,2,\dots,n)$.
We say that a function $\Phi\colon \mathbb{R}^N\to\mathbb{R}$ is strictly \emph{convex} if 
\begin{align}
\lambda\Phi(z)+(1-\lambda)\Phi(z')> \Phi(\lambda z+(1-\lambda)z')
\end{align}
for all $z,z'\in\mathbb{R}^N$ and $\lambda\in(0,1)$.
A typical example of symmetric strictly convex functions is the square-sum $\Phi(z)=\sum_{i\in N}z_i^2$.
In general, for $z\in\mathbb{R}^N$ and $i,j\in N$ with $z_i>z_j$, we have $\Phi(z-\epsilon(\chi_i-\chi_j))<\Phi(z)$ for any $\epsilon\in(0,z_i-z_j)$ because
\begin{align}\begin{split}
    \Phi(z)
    &=\lambda\Phi(z-(z_i-z_j)(\chi_i-\chi_j))+(1-\lambda)\Phi(z)\\
    &>\Phi(\lambda(z-(z_i-z_j)(\chi_i-\chi_j))+(1-\lambda)z)
    =\Phi(z-\lambda(z_i-z_j)(\chi_i-\chi_j)) \label{eq:ssc}
\end{split}\end{align}
for any $\lambda\in(0,1)$.
Here, $\chi_i$ represents a unit vector where only the $i$th component is equal to 1, while all other components are equal to 0.
An allocation $\pi$ is called \emph{$\Phi$-fair} if the utility vector $(\pi_1(E),\dots,\pi_n(E))$ minimizes $\Phi$ among allocations.

Our problem is to find a $\Phi$-fair allocation. 
When there are only divisible goods or only indivisible goods, it suffices to find a utility vector $z^*$ that minimizes $\Phi$, because once we obtain $z^*$, a $\Phi$-fair allocation is obtained from the maximum flow problem as shown below.
Let $G=(\{s,t\}\cup E\cup N, A)$ be a directed graph, where $A=\{(s,e) \mid e\in E\}\cup \{(e,i) \mid v_{i,e}=1\} \cup \{(i,t)\mid i\in N\}$.
The capacity $c$ is defined as $c(s,e)=1$ for $e\in E$, $c(e,i)=1$ for $(e,i)\in A\cap (E\times N)$, and $c(i,t)=z^*_i$ for $i\in N$.
For a maximum flow $f$ from $s$ to $t$, 
let $\pi$ be the allocation defined as $\pi_{ie}=f(e,i)$ ($e\in E, i\in N$).
Then, we can see that $\pi_i(E)=z^*_i$ for each agent $i\in N$.
Since we can find an integral maximum flow in the indivisible goods case, $\pi$ can be an integral allocation.
However, when both types of goods exist, it is not straightforward to construct an allocation from a given utility vector.
Indeed, checking the existence of an allocation achieving a given utility vector is NP-hard, as shown in Appendix (see \Cref{thm:util-hard}).

As we will see in the next section, the set of possible utility vectors forms an integral base-polyhedron (or M-convex set) if there are only divisible (or indivisible) goods.
By using the properties known for integral base-polyhedrons (M-convex sets), we can prove that the dec-min allocations coincide with the inc-max allocations and the $\Phi$-fair allocations for any symmetric strictly convex function $\Phi$.
\begin{theorem}[Fujishige~\cite{fujishige1980} and Maruyama~\cite{Maruyama1978}\footnotemark]\label{thm:divisible}
If there are only divisible goods (i.e., $M=\emptyset$), the dec-min allocations have the same utility vector up to value equivalence. 
Also, the utility vector is the unique utility vector for inc-max allocations and $\Phi$-fair allocations for any symmetric strictly convex functions $\Phi$.
\end{theorem}
\footnotetext{See also Nagano~\cite[Corollary 13]{nagano2007}.}

\begin{theorem}[Frank and Murota~\cite{FM2022a}]\label{thm:indivisible}
If there are only indivisible goods (i.e., $C=\emptyset$), the set of utility vectors of the dec-min allocations forms an M-convex set. 
The set is identical for the set of utility vectors of inc-max allocations or $\Phi$-fair allocations for any symmetric strictly convex functions $\Phi$.
\end{theorem}
However, this is not the case when both types of goods exist (i.e., $M\ne\emptyset$ and $C\ne\emptyset$).
\begin{example}\label{ex:incmax=SOS}
Consider an instance with five agents $N=\{1,2,3,4,5\}$, five indivisible goods $M=\{g_1,g_2,g_3,g_4,g_5\}$, and three divisible goods $C=\{c_1,c_2,c_3\}$.
Suppose that agents $1$, $2$, $3$, and $4$ desire all the goods, but agent $5$ desires only the indivisible goods (see Table~\ref{tbl:incmax=SOS}).
Then, an allocation $\pi$ with $\pi(E)=(7/4,7/4,7/4,7/4,1)$ is dec-min.
However, an allocation $\rho$ with $\rho(E)=(6/4,6/4,6/4,6/4,2)$ is inc-max and square-sum minimizer.
Indeed, $\sum_{i\in N}\pi_i(E)^2=13.25$ and $\sum_{i\in N}\rho_i(E)^2=13$.
\end{example}

\begin{example}\label{ex:decmin=SOS}
Consider an instance with five agents $N=\{1,2,3,4,5\}$, five indivisible goods $M=\{g_1,g_2,g_3,g_4,g_5\}$, and \emph{two} divisible goods $C=\{c_1,c_2\}$.
Suppose that agents $1$, $2$, $3$, and $4$ desire all the goods, but agent $5$ desires only the indivisible goods (see Table~\ref{tbl:decmin=SOS}).
Then, an allocation $\pi$ with $\pi(E)=(6/4,6/4,6/4,6/4,1)$ is dec-min and square-sum minimizer.
However, an allocation $\rho$ with $\rho(E)=(5/4,5/4,5/4,5/4,2)$ is inc-max.
Indeed, $\sum_{i\in N}\pi_i(E)^2=10$ and $\sum_{i\in N}\rho_i(E)^2=10.25$.
\end{example}

\begin{table}[htb]
  \begin{minipage}{.5\textwidth}
    \centering
    \caption{Agents' valuations in Example~\ref{ex:incmax=SOS}.}\label{tbl:incmax=SOS}
    \begin{tabular}{c|cccccccc}
      \toprule
      agents & $g_1$ & $g_2$ & $g_3$ & $g_4$ & $g_5$ & $c_1$ & $c_2$ & $c_3$\\\midrule
      $1$    & $1$   & $1$   & $1$   & $1$   & $1$   & $1$   & $1$   & $1$  \\
      $2$    & $1$   & $1$   & $1$   & $1$   & $1$   & $1$   & $1$   & $1$  \\
      $3$    & $1$   & $1$   & $1$   & $1$   & $1$   & $1$   & $1$   & $1$  \\
      $4$    & $1$   & $1$   & $1$   & $1$   & $1$   & $1$   & $1$   & $1$  \\
      $5$    & $1$   & $1$   & $1$   & $1$   & $1$   & $0$   & $0$   & $0$  \\
      \bottomrule
    \end{tabular}
  \end{minipage}%
  \begin{minipage}{.5\textwidth}
    \centering
    \caption{Agents' valuations in Example~\ref{ex:decmin=SOS}.}\label{tbl:decmin=SOS}
    \begin{tabular}{c|ccccccc}
      \toprule
      agents & $g_1$ & $g_2$ & $g_3$ & $g_4$ & $g_5$ & $c_1$ & $c_2$ \\\midrule
      $1$    & $1$   & $1$   & $1$   & $1$   & $1$   & $1$   & $1$ \\
      $2$    & $1$   & $1$   & $1$   & $1$   & $1$   & $1$   & $1$ \\
      $3$    & $1$   & $1$   & $1$   & $1$   & $1$   & $1$   & $1$ \\
      $4$    & $1$   & $1$   & $1$   & $1$   & $1$   & $1$   & $1$ \\
      $5$    & $1$   & $1$   & $1$   & $1$   & $1$   & $0$   & $0$ \\
      \bottomrule
    \end{tabular}
  \end{minipage}
\end{table}

Nevertheless, a certain symmetric strictly convex function $\Phi$ induces the dec-min and inc-max solution as a $\Phi$-fair allocation. 
\begin{proposition}\label{prop:Phi}
Let $(N,M,C,v)$ be a fair allocation instance.
There exists a symmetric strictly convex function $\Phi$ such that a dec-min allocation is $\Phi$-fair.
In addition, there exists a symmetric strictly convex function $\Phi'$ such that an inc-max allocation is $\Phi'$-fair.
\end{proposition}
\begin{proof}

When we fix an allocation of the indivisible goods, we can observe that the utility vector of the dec-min allocations is uniquely determined by treating each indivisible good as a divisible one such that only the agent who receives it desires it and applying Theorem~\ref{thm:divisible}.
Let $U$ be the set of the utility vectors obtained in this way for all possible allocations of indivisible goods.
Since the number of possible allocations of indivisible goods is finite, $|U|$ is also finite.

Let $x$ be the utility vector of a dec-min allocation.
Without loss of generality, we may assume that $y^\downarrow\ne x^\downarrow$ for some $y\in U$ since otherwise the claim of this proposition holds for any symmetric strictly convex function $\Phi$.
Let $\epsilon=\min\{y^\downarrow_i-x^\downarrow_i\mid y\in U,~i\in [n],~y^\downarrow_i>x^\downarrow_i\}$.
We show that $x$ minimizes $\Phi(z)=\sum_{i\in N}(2n)^{z_i/\epsilon}$ over $U$.
Take any $y\in U$ with $y^\downarrow\ne x^\downarrow$.
Then, there exists an index $k\in[n]$ such that $y^\downarrow_i=x^\downarrow_i$ for $i=1,2,\dots,k-1$ and $y^\downarrow_k\ge x^\downarrow_k+\epsilon$ since $x$ is the utility vector of a dec-min allocation.
Hence, we obtain
\begin{align}
    \Phi(y)-\Phi(x)
    =\sum_{i\in N}(2n)^{y_i/\epsilon}-\sum_{i\in N}(2n)^{x_i/\epsilon}
    \ge (2n)^{y^\downarrow_k/\epsilon}-n\cdot (2n)^{(y^\downarrow_k-\epsilon)/\epsilon}
    = \frac{1}{2}\cdot (2n)^{y^\downarrow_k/\epsilon}>0.
\end{align}
The inc-max case is also proven in a similar way by setting $\Phi'(z)=\sum_{i\in N}(2n)^{-z_i/\epsilon'}$ with $\epsilon'=\min\{x_i^\downarrow-y_i^\downarrow\mid y\in U,\ i\in [n],\ x_i^\downarrow>y_i^\downarrow\}$.
\end{proof}
This proposition implies that dec-min and inc-max allocations can be viewed as minimizers of some symmetric strictly convex function.
We remark that some prominent fairness notions are naturally represented as $\Phi$-fairness for some symmetric strictly convex function $\Phi$.
An allocation $\pi$ is said to achieve the \emph{maximum Nash welfare} (MNW) if the number of agents with positive utilities is maximized, and subject to that, the Nash welfare $\prod_{i\in N: \pi_i(E)>0} \pi_i(E)$ is maximized.
Finding a utility vector of an MNW allocation is equivalent to minimizing $\prod_{i\in N: z_i>0} (z_i + \varepsilon)$ for some sufficiently small $\varepsilon>0$.
See Appendix~\ref{sec:MNW} for the detail.
The egalitarian social welfare is defined by the smallest utility among agents.
Maximizing the egalitarian social welfare is a weaker notion of inc-max, and thus we can apply \Cref{prop:Phi}.
Hence, we mainly consider $\Phi$-fair allocations in the following.

\section{Integral base-polyhedra and M-convex sets}\label{sec:M-convex}
In this section, we formulate our fair allocation problem in terms of convex analysis. 
Then we introduce known results and tools for continuous or discrete cases.

\subsection{Definitions and properties}\label{subsec:definitions}
In this subsection, we introduce the definitions and properties of integral base-polyhedra and M-convex sets that we will use later.
A set function $f$ over $N$ is called \emph{supermodular} if 
\begin{align}
    f(X)+f(Y)\le f(X\cup Y)+f(X\cap Y) \quad (\forall X,Y\subseteq N)
\end{align}
and \emph{submodular} if 
\begin{align}
    f(X)+f(Y)\ge f(X\cup Y)+f(X\cap Y) \quad (\forall X,Y\subseteq N).
\end{align}
For a subset $X\subseteq N$ and a vector $x \in \mathbb{R}^N$, we denote $x(X) = \sum_{i\in X}x_i$.
For an integer-valued supermodular set function $f$ on $N$ for which $f(\emptyset)=0$ (normalized), 
the \emph{integral base-polyhedron} $\bar{\B}$ of $f$ is defined as
\begin{align}
    \bar{\B}=\left\{x\in\mathbb{R}^N\mid \textstyle x(N)=f(N)\text{ and } x(X)\ge f(X)~(\forall X\subseteq N)\right\}.
\end{align}
In addition, we call the set $\ddot{\B}$ of the integer vectors in an integral base-polyhedron $\bar{\B}$ an \emph{M-convex set}.
Note that an M-convex set $\ddot{\B}$ induces an integral base-polyhedron $\bar{\B}$ as its convex hull.

For an integral base-polyhedra $\bar{\B}$ defined by a supermodular function $f\colon 2^N\to\mathbb{Z}$, we have the following facts.
\begin{proposition}[{\cite[Lemma 3.2]{fujishige2005}}]\label{prop:face}
For any $X\subseteq N$, there exists $x\in\ddot{\B}$ such that $x(X)=f(X)$.
\end{proposition}

\begin{proposition}[{\cite[Lemma 2.2]{fujishige2005}}]\label{prop:tightsets}
For $x\in\bar{\B}$, letting $\mathcal{D}(x)=\{X\subseteq N\mid x(X)=f(X)\}$, we have
\begin{align}
    X,Y\in\mathcal{D}(x) \implies X\cup Y,\,X\cap Y\in\mathcal{D}(x).
\end{align}
\end{proposition}

An M-convex set $\ddot{B}$ satisfies the \emph{exchange property} (B-EXC[$\mathbb{Z}$]) stated as follows.
\begin{proposition}[\cite{Murota:DCA}]\label{prop:B-EXC-Z}
For any $x,y\in \ddot{\B}$ and $i\in\supp^+(x-y)$, there exists some $j\in\supp^-(x-y)$ such that $x-\chi_i+\chi_j\in \ddot{\B}$ and $y+\chi_i-\chi_j\in \ddot{\B}$.
\end{proposition}

An integral base-polyhedra $\bar{\B}$ satisfies the following exchange property (B-EXC[$\mathbb{R}$]).
\begin{proposition}[\cite{Murota:DCA}]\label{prop:B-EXC-R}
For any $x,y\in \bar{\B}$ and $i\in\supp^+(x-y)$, there exists some $j\in\supp^-(x-y)$ and a positive real $\alpha_0$ such that $x-\alpha(\chi_i-\chi_j)\in \bar{\B}$ and $y+\alpha(\chi_i-\chi_j)\in \bar{\B}$ for all $\alpha\in[0,\alpha_0]$.   
\end{proposition}

In addition, the following variants of exchange properties also hold.
Suppose that $\bar{\B}$ and $\ddot{\B}$ are induced by a supermodular function $f\colon 2^N\to\mathbb{Z}$.
\begin{proposition}\label{prop:B-EXC-XZ}
For any $x,y\in \ddot{\B}$ and $X\subseteq N$ with $x(X)>y(X)$, there exist $i\in X$ and $j\in N\setminus X$ such that $x-\chi_i+\chi_j\in \ddot{\B}$.
Moreover, if $x(X)>f(X)$, there exist $i\in X$ and $j\in N\setminus X$ such that $x-\chi_i+\chi_j\in \ddot{\B}$.
\end{proposition}
\begin{proof}
Suppose to the contrary that $x-\chi_i+\chi_j\not\in\ddot{\B}$ for any $i\in X$ and $j\in N\setminus X$.
For each $i\in X$ and $j\in N\setminus X$, let $Y_{ij}$ be a subset of $N$ such that $x(Y_{ij})=f(Y_{ij})$, $i\in Y_{ij}$, and $j\not\in Y_{ij}$.
Then, $Y_i\coloneqq\bigcap_{j\in N\setminus X}Y_{ij}$ satisfies $i\in Y_i$, $Y_i\subseteq X$, and $x(Y_i)=f(Y_i)$ for each $i\in N$ by \Cref{prop:tightsets}.
Furthermore, $Y\coloneqq\bigcup_{i\in X}Y_i$ satisfies $Y=X$ and $x(Y)=f(X)$ by \Cref{prop:tightsets}.
Hence, we obtain $f(X)=x(X)>y(X)\ge f(X)$, which is a contradiction.

In addition, the latter statement holds by setting $y$ such that $y(X)=f(X)$, whose existence is guaranteed by \Cref{prop:face}.
\end{proof}

\begin{proposition}\label{prop:B-EXC-XR}
For any $x,y\in \bar{\B}$ and $X\subseteq N$ with $x(X)>y(X)$, there exist $i\in X$, $j\in N\setminus X$, and $\epsilon>0$ such that $x-\epsilon(\chi_i-\chi_j)\in \bar{\B}$.
Moreover, if $x(X)>f(X)$, there exist $i\in X$, $j\in N\setminus X$, and $\epsilon>0$ such that $x-\epsilon(\chi_i-\chi_j)\in \bar{\B}$.
\end{proposition}
\begin{proof}
The proof is done in a similar way to \Cref{prop:B-EXC-XZ}.
\end{proof}

\begin{proposition}\label{prop:B-EXC-IJZ}
For any $x\in \ddot{\B}$ and disjoint subsets $I,J\subseteq N$ such that $x(J)<f(I\cup J)-f(I)$, there exist $i\in I$ and $j\in J$ such that $x-\chi_i+\chi_j\in \ddot{\B}$.
\end{proposition}
\begin{proof}
Since $f(I\cup J)\le x(I\cup J)< x(I)+f(I\cup J)-f(I)$, we have $x(I)>f(I)$.

Suppose to the contrary that $x-\chi_i+\chi_j\not\in\ddot{\B}$ for any $i\in I$ and $j\in J$.
For each $i\in I$ and $j\in J$, let $T_{ij}$ be a subset of $N$ such that $x(Y_{ij})=f(T_{ij})$, $i\in T_{ij}$, and $j\not\in T_{ij}$.
Then, $T\coloneqq\bigcup_{i\in I}\bigcap_{j\in J}T_{ij}$ satisfies $I\subseteq T\subseteq N\setminus J$ and $x(T)=f(T)$ by \Cref{prop:tightsets}.

By \Cref{prop:B-EXC-XZ} and $x(I)>f(I)$, there exist $i_1\in I$ and $j_1\in N\setminus I$ such that $x^{(1)}\coloneqq x-\chi_{i_1}+\chi_{j_1}\in\ddot{\B}$.
Here, $j_1$ must be in $T$ since otherwise $x^{(1)}(T)=x(T)-1 \geq f(T)$, which is a contradiction.
Let $s \coloneqq f(I)-x(I)$.
Similarly, for each $k=1,2,\dots,s$, 
since $x^{(k-1)}(I)>f(I)$,
we can choose $i_k\in I$ and $j_k\in T\setminus I$ such that $x^{(k)}\coloneqq x^{(k-1)}-\chi_{i_k}+\chi_{j_k}\in\ddot{\B}$.
However, we have $x^{(s)}(I\cup J)=f(I)+x(J)<f(I\cup J)$, which contradicts $x^{(s)}\in\ddot{\B}$.
\end{proof}

\begin{proposition}\label{prop:B-EXC-M}
For any $x\in \ddot{\B}$ and $y\in\bar{\B}$, and $i\in \supp^+(x-y)$, there exists $j\in\supp^-(x-y)$ such that $x-\chi_i+\chi_j\in\ddot{\B}$.
Also, for any $x\in \ddot{\B}$ and $y\in\bar{\B}$, and $i\in \supp^-(x-y)$, there exists $j\in\supp^+(x-y)$ such that $x+\chi_i-\chi_j\in\ddot{\B}$.
\end{proposition}
\begin{proof}
We only provide a proof for the former part, as the latter part can be demonstrated in a similar manner.
By \Cref{prop:B-EXC-R}, there exists some $j\in\supp^-(x-y)$ and a positive real $\alpha_0$ such that $x'\coloneqq x-\alpha_0(\chi_i-\chi_j)\in \bar{\B}$.
For any $X\subseteq N$ with $i\in X$ and $j\not\in X$, we have $x'(X)=x(X)-\alpha_0\ge f(X)$, and hence $x(X)-1\ge f(X)$ since $x(X)$ and $f(X)$ are integers.
Therefore, $x-\chi_i+\chi_j \in \ddot{B}$ holds.
\end{proof}

Moreover, the following properties also hold.
\begin{proposition}\label{prop:M-ijk}
    For any $x\in\ddot{\B}$, if 
    $x-\chi_{i}+\chi_{j}\in\ddot{\B}$ and 
    $x-\chi_{j}+\chi_{k}\in\ddot{\B}$,
    then it holds that $x-\chi_i+\chi_k\in\ddot{\B}$.
\end{proposition}
\begin{proof}
    Let $x'\coloneqq x-\chi_{i}+\chi_{j}$ and $y'\coloneqq x-\chi_{j}+\chi_{k}$. 
    We apply \Cref{prop:B-EXC-Z} to $x'$ and $y'$. 
    Since $\supp^+(x'-y')=\{j\}$ and $\supp^-(x'-y')=\{i,k\}$, we have two possibilities.
    In both cases, we have $x'-\chi_j+\chi_k \in \ddot{B}$.
\end{proof}

\begin{proposition}[{\cite[Lemma~4.5]{fujishige2005}}]\label{prop:no-shortcut}
    Let $x\in\ddot{\B}$. 
    Suppose that there exist a sequence $i_1,j_1,\dots,i_r,j_r$ of $2r$ distinct elements in $N$
    such that 
    $x-\chi_{i_h}+\chi_{j_k}\in\ddot{\B}$ if $h=k$ and 
    $x-\chi_{i_h}+\chi_{j_k}\not\in\ddot{\B}$ if $h>k$ for $h,k\in[r]$. 
    Then, it holds that $x-\sum_{k\in[r]}(\chi_{i_k}-\chi_{j_k})\in\ddot{\B}$.
\end{proposition}

\begin{proposition}[integrally convexity~\cite{Murota:DCA}]\label{prop:integrally-convexity}
\begin{align}
x\in\bar{\B} \implies x\in\conv(\bar{\B}\cap \{y\in\mathbb{Z}^N\mid \|x-y\|_{\infty}<1\}).
\end{align}    
\end{proposition}

\subsection{The hybrid problem and relation to fair allocation}\label{subsec:relation_fair}
Here, we explain how integral base-polyhedra and M-convex sets appear in a fair allocation with binary additive valuations.
For a subset $X\subseteq N$ of agents, we define $f_M,f_C,f_E\colon 2^N\to\mathbb{Z}_+$ as follows:
\begin{itemize}
    \item $f_M(X)=|\{g\in M: v_{ig}=0~(\forall i\not\in X)\}|$ is the number of indivisible goods that must be allocated to agents in $X$,
    \item $f_C(X)=|\{c\in C: v_{ic}=0~(\forall i\not\in X)\}|$ is the number of divisible goods that must be allocated to agents in $X$, and
    \item $f_E(X)=f_M(X)+f_C(X)$ is the number of goods that must be allocated to agents in $X$.
\end{itemize}
It is not difficult to see that the functions $f_M,f_C,f_E$ are normalized integer-valued supermodular.\footnote{Most of our results can be extended to the case when each agent evaluates indivisible goods with a matroid rank function and divisible goods with the concave closure of a matroid rank function because the functions $f_M,f_C,f_E$ continue to be normalized integer-valued supermodular in this scenario. }
Let $\ddot{\B}_M$ and $\bar{\B}_C$ be the M-convex set of $f_M$ and the integral base-polyhedron of $f_C$, respectively. In addition, let $\B_E$ be the Minkowski sum of $\ddot{\B}_M$ and $\bar{\B}_C$, i.e., $\B_E=\{x+y\mid x\in\ddot{\B}_M,\ y\in\bar{B}_C\}$.
Then, $\B_E$ is the set of possible utility vectors.

Recall that $\bar{\B}_M=\conv(\ddot{\B}_M)$ and  $\ddot{\B}_C=\bar{B}_C\cap\mathbb{Z}^N$.
We denote $\bar{\B}_E=\conv(\B_E)$, and $\ddot{\B}_E=\B_E\cap\mathbb{Z}^N$.
Note that $\B_E$ is not necessarily an M-convex set or an integral base-polyhedron as we have seen in \Cref{ex:neither}.

Then we can rewrite finding a $\Phi$-fair allocation for a symmetric strictly convex function $\Phi\colon \mathbb{R}^N\to\mathbb{R}$ as the problem of finding a vector $z$ that attains
\begin{align}
\min_{z\in\B_E}\Phi(z)~(=\min_{x\in\ddot{\B}_M}\min_{y\in\bar{\B}_C}\Phi(x+y)).
\end{align}
This formulation is not restricted to the fair allocation.
An optimal solution $z$ is called \emph{$\Phi$-minimizer} (on $\B_E$).
We omit the feasible region if it is clear.
If $M=\emptyset$ or $C=\emptyset$, this minimization problem can be solved in polynomial time, as shown in the next subsection.

\subsection{Principal Partition and Canonical Partition}\label{subsec:partition}

In this subsection, we summarize the structures of $\Phi$-minimizers on integral base-polyhedra or M-convex sets.

Consider the integral base-polyhedron $\bar{\B}$ and the M-convex set $\ddot{\B}$ of a supermodular function $f\colon 2^N\to\mathbb{Z}_+$.
For any real number $\lambda$, let $\cL(\lambda)$ be the set of all maximizers of $f(X)-\lambda|X|$, i.e., $\cL(\lambda)=\argmax_{X\subseteq N}(f(X)-\lambda|X|)$. 
Note that $\cL$ has a lattice structure (see \Cref{prop:tightsets}), i.e., $\cL$ is closed under union and intersection. Let $\rL(\lambda)$ be the smallest member in $\cL(\lambda)$.
It is known that $\rL(\lambda)\subseteq \rL(\lambda')$ for any $\lambda>\lambda'$ (see, e.g., \cite[Proposition 3.1]{FMdecmin2}).

Fujishige~\cite{fujishige1980} characterized the optimal utility vectors by the \emph{principal partition} of $N$.
There are at most $|N|$ number of $\lambda$ for which $|\cL(\lambda)|\ge 2$. Let us denote such numbers as $\lambda_1>\lambda_2>\dots>\lambda_r$, which are called the \emph{critical values}.
The principal partition $\hat{N}_1,\hat{N}_2,\dots,\hat{N}_r$ is a partition of $N$ defined by
\begin{align}
    \hat{N}_j=\rL(\lambda_j')-\rL(\lambda_j)\quad(j=1,2,\dots,r),
\end{align}
where $\lambda_j'$ is an arbitraly real satisfying $\lambda_j>\lambda_j'>\lambda_{j+1}$ (assuming that $\lambda_{r+1}=-\infty$).
\begin{theorem}[{Fujishige~\cite{fujishige1980} and Maruyama~\cite{Maruyama1978}}]\label{thm:principal}
    The unique minimizer $x^*$ of $\min_{x\in\bar{\B}}\Phi(x)$ satisfies $x^*_i=\lambda_j$ for each $i\in\hat{N}_j$ and $j \in [r]$.
    The principal partition and critical values can be found in strongly polynomial time.
\end{theorem}
For more details of the principal partition, see a book and a survey of Fujishige~\cite{fujishige2005,fujishige2009}.

Frank and Murota~\cite{FM2022a} characterized the optimal utility vectors by the \emph{canonical partition} of $N$.
There are at most $|N|$ number of $\beta\in\mathbb{Z}$ for which $\rL(\beta)\ne\rL(\beta-1)$. 
Let us denote such numbers as $\beta_1>\beta_2>\dots>\beta_q$, which are called the \emph{essential values}.
The canonical partition $N_1,N_2,\dots,N_q$ is a partition of $N$ defined by
\begin{align}
    N_i=\rL(\beta_i-1)-\rL(\beta_i)\quad(i=1,2,\dots,q).
\end{align}
Alternatively, the canonical partition and the essential values can be obtained by the following procedure~\cite[Section 3]{FMdecmin2}: for $j=1,2,\dots,q$, define
\begin{align}
    \beta_j&=\max\left\{\left\lceil\frac{f(X\cup \bigcup_{j'=1}^{j-1}N_{j'})-f(\bigcup_{j'=1}^{j-1}N_{j'})}{|X|}\right\rceil\mid \emptyset\ne X\subseteq N\setminus \bigcup_{j'=1}^{j-1}N_{j'}\right\},\label{eq:beta_j}\\
    h_j(X)&=\textstyle f(X\cup \bigcup_{j'=1}^{j-1}N_{j'})-(\beta_{j}-1)|X|-f(\bigcup_{j'=1}^{j-1}N_{j'})\quad(\forall X\subseteq N\setminus \bigcup_{j'=1}^{j-1}N_{j'}), \ \text{and}\\
    N_j&=\textstyle\text{smallest subset of $N\setminus \bigcup_{j'=1}^{j-1}N_{j'}$ maximizing }h_j.
\end{align}
They provided a strongly polynomial-time algorithm to compute the canonical partition and the essential values by using this structure and a strongly polynomial-time algorithm for the submodular function minimization~\cite{IFF2001,Schrijver2000}.
\begin{theorem}[{Frank and Murota~\cite{FM2022a,FM2022b,FMdecmin2}}]\label{thm:canonical}
    The essential values $\beta_1>\beta_2>\dots>\beta_q$ are obtained from the critical values $\lambda_1>\lambda_2>\dots>\lambda_r$ as the distinct members of the rounded-up integers $\lceil\lambda_1\rceil\ge\lceil\lambda_2\rceil\ge\dots\ge\lceil\lambda_r\rceil$.
    Moreover, the canonical partition is an aggregation of the principal partition as $N_i=\bigcup_{j:\,\lceil \lambda_j\rceil=\beta_i}\hat{N}_j$ for each $i\in[q]$.
    Any minimizer $y^*$ of $\min_{y\in\ddot{\B}}\Phi(y)$ satisfies $\beta_j-1\le y^*_i\le\beta_j$ for each $i\in N_j$ and $j\in [q]$.
    The minimizer $y^*$, the canonical partition, and essential values can be found in strongly polynomial time with respect to $|N|$.
\end{theorem}

\section{Structure of \texorpdfstring{$\Phi$}{Phi}-minimizers}\label{sec:structure}

In this section, we prove a proximity theorem and a structure based on the canonical partition.
Let $f_M$ and $f_C$ be two supermodular functions over $M$ and $C$, respectively.
Let also $f_E=f_M+f_C$. 
Recall that $\ddot{\B}_M$ and $\bar{\B}_C$ are a corresponding M-convex set and integral base-polyhedron, respectively.
In addition, $\B_E$ is the Minkowski sum of $\ddot{\B}_M$ and $\bar{\B}_C$.
Note that, in this section, we consider general integral supermodular functions, not restricted to the ones appearing in Section~\ref{subsec:relation_fair}.

\begin{theorem}\label{thm:proximity}
Let $\Phi$ be a symmetric strictly convex function.
For any $z^*\in\argmin_{z\in\B_E}\Phi(z)$ and $\bar{z}\in\argmin_{z\in\bar{\B}_E}\Phi(z)$, 
we have
\begin{align}
\lfloor \bar{z}_i\rfloor \le z^*_i\le \lceil \bar{z}_i\rceil \quad (i\in N).
\end{align}
\end{theorem}

It should be noted that, for the integral base-polyhedron $\bar{\B}$ and the M-convex set $\ddot{\B}$ of a common supermodular function, the following proximity theorem has been shown by Frank and Murota~\cite{FMdecmin2}.
\begin{theorem}[{\cite[Theorem 4.1]{FMdecmin2}}]\label{thm:proximity_FM}
Let $\Phi$ be a symmetric strictly convex function.
For any $x^*\in\argmin_{x\in\ddot{\B}}\Phi(x)$ and $y^*\in\argmin_{y\in\bar{\B}}\Phi(y)$, 
we have
\begin{align}
\lfloor y^*_i\rfloor \le x^*_i\le \lceil y^*_i\rceil \quad (i\in N).
\end{align}
\end{theorem}
Note that \Cref{thm:proximity_FM} is a special case of our 
\Cref{thm:proximity} when $f_C(X)=0~(\forall X\subseteq N)$.
We prove \Cref{thm:proximity} following the same approach as for \Cref{thm:proximity_FM}. However, we need to conduct a more detailed analysis to handle $\bar{\B}_C$.

Throughout this section, we fix a symmetric strictly convex function $\Phi\colon\mathbb{R}^N\to\mathbb{R}$ and its minimizer $z^*\in\argmin_{z\in\B_E}\Phi(z)$.
By definition, $z^*$ can be represented as $x^*+y^*$ by $x^*\in\ddot{\B}_M$ and $y^*\in\bar{\B}_C$.
Moreover, let $N_1,\dots,N_q$ and $\beta_1,\dots,\beta_q$ be the canonical partition and the essential values of the M-convex set $\ddot{\B}_E$.

In subsequent subsections, we prove \Cref{thm:proximity} through the following steps.
First, in Section~\ref{sec:proximity-peak}, we demonstrate that $z_i^*$ lies within the interval $[\beta_1-1,\, \beta_1]$ for $i\in N_1$.
Then, in Section~\ref{sec:proximity-decomp}, we decompose the problem of finding $z^*$ into two independent problems on $N_1$ and $N\setminus N_1$. 
By iteratively applying the same procedure, we obtain the desired structure.
In Section~\ref{sec:proximity-fair}, we show an additional result when $\B_E$ emerges from the fair allocation.

\subsection[The peak-set]{The peak-set $N_1$}\label{sec:proximity-peak}
In this subsection, we prove that $\beta_1\ge z_i^*\ge \beta_1-1$ and $z^*(N_1)=f_E(N_1)$.

We first observe that we can transfer some amounts from elements with high value to those with low value.
\begin{lemma}\label{lem:div-layer}
    If $y^*+\epsilon(\chi_i-\chi_j)\in\bar{\B}_C$ for some $i,j\in N$ and $\epsilon>0$, then $z^*_i=x^*_i+y^*_i\ge x^*_j+y^*_j=z^*_j$.
    In addition, for any $\beta\in\mathbb{R}$, it holds that $y^*(N')=f_C(N')$ with $N'=\{i\in N\mid z^*_i\ge \beta\}$.
\end{lemma}
\begin{proof}
    For the former statement, suppose to the contrary that $y^*+\epsilon(\chi_i-\chi_j)\in\bar{\B}_C$ for some $i,j\in N$ such that $x^*_i+y^*_i<x^*_j+y^*_j$ and $\epsilon>0$.
    Then, $y=y^*+\min\{\epsilon,\,(y^*_j-y^*_i)/2\}\cdot(\chi_i-\chi_j)\in\bar{\B}_C$ and $\Phi(x^*+y)<\Phi(x^*+y^*)$ by \eqref{eq:ssc}.
    This contradicts the assumption that $z^*=x^*+y^*$ is a $\Phi$-minimizer.
    
    For the latter statement, suppose that $y^*(N')>f_C(N')$ for some $N'=\{i\in N\mid z^*_i\ge \beta\}$ with $\beta\in\mathbb{R}$.
    Then, by \Cref{prop:B-EXC-XR}, there exist $i\in N'$, $j\in N\setminus N'$, and $\epsilon>0$ such that $y^*+\epsilon(\chi_j-\chi_i)\in\bar{\B}_C$.
    By the former statement, this implies $z^*_j\ge z^*_i$, which contradicts $j\not\in N'$.
\end{proof}

\begin{lemma}\label{lem:transfer}
    Let $\beta$ be an integer and $N'\subseteq N$.
    Define 
$N^>=\{i\in N \mid z^*_i >\beta \}$, $N^==\{i\in N \mid z^*_i =\beta \}$, and $N^<=\{i\in N \mid z^*_i <\beta \}$.
    Construct a graph
    \begin{align}
G=\big(N',\,\big\{(i,j)\in N'\times N'\mid x^*-\chi_i+\chi_j\in\ddot{\B}_M~\text{or}~y^*-\chi_i+\chi_j\in\bar{\B}_C\big\}\big).
\end{align}
    If $G$ has a path from some $i\in N'\cap N^>$ to some $j \in N'\cap N^>$,
    then there exists a vector $z''$ such that $\Phi(z'') < \Phi(z^*)$.
\end{lemma}
\begin{proof}
We first observe that Lemma~\ref{lem:div-layer} implies $y^*(N^>)=f_C(N^>)$
and $y^*(N^>\cup N^=)=f_C(N^>\cup N^=)$.

    Let $P=(i_1,i_2,\dots,i_k)$ be a shortest path from some $i_1 \in N'\cap N^>$ to some $i_k \in N'\cap N^<$.
    Then we have $i_1\in N^>$, $i_2,\dots,i_{k-1}\in N^=$, and $i_k\in N^<$.
    Since $z^*_{i_2} < z^*_{i_1}$, \Cref{lem:div-layer} implies that $y^*-\chi_{i_{1}}+\chi_{i_{2}}\notin\bar{\B}_C$, and thus $x^*-\chi_{i_1}+\chi_{i_{2}}\in\ddot{\B}_M$.
If $x^*-\chi_{i_2}+\chi_{i_{3}}\in\ddot{\B}_M$, then $x^*-\chi_{i_1}+\chi_{i_{3}}\in\ddot{\B}_M$ by \Cref{prop:M-ijk}, which leads to a shortcut of $P$.
Thus, we have $x^*-\chi_{i_2}+\chi_{i_{3}}\notin\ddot{\B}_M$ and instead, $y^*-\chi_{i_2}+\chi_{i_{3}}\in\bar{\B}_C$ holds.
Next, let us assume that $y^*-\chi_{i_3}+\chi_{i_{4}}\in\bar{\B}_C$ and $i_4 \in N^=$, and derive a contradiction.
Consider 
\begin{align}
\bar{\B}^= = \{ \tilde{y}\in \mathbb{R}^{N^=} \mid \tilde{y}(N^=)=f_C(N^=\cup N^>)-f_C(N^>)\text{ and } \tilde{y}(X) \geq f_C(X\cup N^>)-f_C(N^>) \ (\forall X\subseteq N^=)\},
\end{align}
which is also an integral base-polyhedron of a supermodular function.
Since $\beta$ is an integer, $y^*_{N^=}$ is also integral.
Thus, the vectors $y^*_{N^=}-\chi_{i_2}+\chi_{i_3}$ and $y^*_{N^=}-\chi_{i_3}+\chi_{i_4}$ are contained in $\ddot{\B}^=$, since $y^*(N^>)=f_C(N^>)$ and $i_2, i_3, i_4 \in N^=$.
Then $y^*_{N^=}-\chi_{i_2}+\chi_{i_4}\in \ddot{\B}^=$ follows from \Cref{prop:M-ijk}, which means that $y^*-\chi_{i_2}+\chi_{i_{4}}\in\bar{\B}_C$\footnotemark.
\footnotetext{We need to check for each $X$ such that $i_2 \in X$ but $i_4 \notin X$. The case when $X \supseteq \{i_2,i_3\}$ follows by $y^*-\chi_{i_3}+\chi_{i_4} \in \bar{\B}_C$, and the case when $i_2 \in X$ but $i_3 \notin X$ follows by $y^*-\chi_{i_2}+\chi_{i_3} \in \bar{\B}_C$.}
However, this implies a shortcut of $P$.
Therefore, if $i_4 \in N^=$, then $y^*-\chi_{i_3}+\chi_{i_{4}}\notin\bar{\B}_C$, and hence 
$x^*-\chi_{i_3}+\chi_{i_{4}}\in\ddot{\B}_M$.
By the same argument, we have $x^*-\chi_{i_{\ell}}+\chi_{i_{\ell+1}}\in\ddot{\B}_M$ if $\ell$ is an odd number and $i_{\ell+1}\in N^=$, and $y^*-\chi_{i_{\ell}}+\chi_{i_{\ell+1}}\in\bar{\B}_C$ if $\ell$ is an even number. 
Note that $k$ must be an even number, because $y^*-\chi_{i_{k-1}}+\chi_{i_{k}}\notin\bar{\B}_C$ follows from \Cref{lem:div-layer} and $z^*_{i_{k}}<z^*_{i_{k-1}}$. 
Moreover, for any integers $\ell$ and $h$ with $\ell\ge h+2$, we have $x^*-\chi_{i_\ell}+\chi_{i_{h}}\not\in\ddot{\B}_M$ and $y^*-\chi_{i_\ell}+\chi_{i_{h}}\not\in\bar{\B}_C$.
Let 
$$
x'=x^*-\sum_{\ell:\,\text{odd}}(\chi_{i_{\ell}}-\chi_{i_{\ell+1}}), \quad 
y'=y^*-\sum_{\ell:\,\text{even}}(\chi_{i_{\ell}}-\chi_{i_{\ell+1}}),
$$
and
$z'=x'+y'$.
By \Cref{prop:no-shortcut}, we have $x'\in\ddot{\B}_M$, $y'\in\bar{\B}_C$, and $z'\in\B_E$.
Note that $z'=z^*-\chi_{i_1}+\chi_{i_k}$, $y'(N^>)=y^*(N^>)~(=f_C(N^>))$ and $y'(N^>\cup N^=)=y^*(N^>\cup N^=)~(=f_C(N^>\cup N^=))$ by the construction.
For notational convenience, we denote $i^*=i_1$ and $j^*=i_k$ in the following.

If $z^*_{i^*}>z^*_{j^*}+1$, then $\Phi(z')<\Phi(z^*)$.
Thus, suppose that $z^*_{i^*}\le z^*_{j^*}+1~(<\beta+1)$. 
In this case, $\beta-1<z'_{i^*}<\beta$ and $\beta<z'_{j^*}<\beta+1$.
By \Cref{prop:integrally-convexity}, $y'$ can be represented by a convex combination of its integral neighbors in $\bar{\B}_C$. 
Let $y'=\sum_{t=1}^r\lambda^{(t)}\cdot y^{(t)}$, where $y^{(t)}\in \ddot{\B}_C\cap \{y\in\mathbb{Z}^N\mid \|y-y'\|<1\}~(\forall t\in[r])$, $\sum_{t=1}^r \lambda^{(t)}=1$, and $\lambda^{(t)}\ge 0~(\forall t\in[r])$.
Define $z^{(t)}=x'+y^{(t)}$ for each $t$.
Thus, we also obtain $z'=\sum_{t=1}^r\lambda^{(t)}\cdot z^{(t)}$.
Note that $z^{(t)}_{i^*} \in \{\beta-1, \beta\}$ and $z^{(t)}_{j^*} \in \{\beta, \beta+1\}$ for each $t$.
In addition, for each $t$, it holds that $y^{(t)}(N^>) = f_C(N^>)$ because $\sum_{t=1}^r \lambda^{(t)}\cdot y^{(t)}(N^>) = y'(N^>) = y^*(N^>)=f_C(N^>)$ and $y^{(t)}(N^>) \geq f_C(N^>)$.
Similarly, we can see that $y^{(t)}(N^>\cup N^=) = f_C(N^>\cup N^=)$ for each $t$.

Let us choose an arbitary $t$ with $z^{(t)}_{i^*}=\beta-1$.
Let 
\begin{align}
\bar{\B}^> = \{ \tilde{y}\in \mathbb{R}^{N^>} \mid \tilde{y}(N^>)=f_C(N^>) \text{ and } \tilde{y}(X)\geq f_C(X) \ (\forall X \subseteq N^>)\}
\end{align}
(the restriction of $\bar{\B}_C$ to $N^>$), and $\ddot{\B}^>$ be the M-convex set induced from $\bar{\B}^>$.
Then it holds that $y^{(t)}_{N^>}\in \ddot{\B}^>$, $y'_{N^>}\in \bar{\B}^>$ and $i^* \in \supp^-(y^{(t)}_{N^>}-y'_{N^>})$.
We apply \Cref{prop:B-EXC-M} to them. 
Then we can choose an index $i^{(t)}\in \supp^+(y^{(t)}_{N^>}-y'_{N^>})$ such that 
$y^{(t)}_{N^>}+\chi_{i^*}-\chi_{i^{(t)}}\in\bar{\B}^>$.
We show that this implies 
$$\hat{y}^{(t)} \coloneqq y^{(t)}+\chi_{i^*}-\chi_{i^{(t)}}\in\bar{\B}_C.$$
Indeed, for any $X$ with $i^*\notin X$ and $i^{(t)} \in X$ (the other cases are trivial), since $y^{(t)}(X\cap N^>)-1 \geq f_C(X\cap N^>)$, we have
\begin{align}
    y^{(t)}(X) = y^{(t)}(X\cap N^>) + y^{(t)}(X\cup N^>) - y^{(t)}(N^>)
    > f_C(X\cap N^>) + f_C(X\cup N^>)-f_C(N^>)
    \geq f_C(X),
\end{align}
which implies that $\hat{y}^{(t)}(X) \geq f_C(X)$.
In addition, we observe that $z^{(t)}_{i^{(t)}} = y^{(t)}_{i^{(t)}} + x'_{i^{(t)}} > y'_{i^{(t)}} + x'_{i^{(t)}} = z'_{i^{(t)}} = z_{i^{(t)}} > \beta$.
Thus, since $z^{(t)}_{i^{(t)}}$ is an integer, 
\begin{align}\label{eq:i^t}
    \hat{y}^{(t)}_{i^{(t)}} + x'_{i^{(t)}} = z^{(t)}_{i^{(t)}}-1 \geq  \beta.
\end{align}
On the other hand, 
for each $t$ with $z^{(t)}_{i^*}=\beta$, we denote $\hat{y}^{(t)} = y^{(t)}$.
We also remark that $\hat{y}^{(t)}(N^>\cup N^=)=y^{(t)}(N^>\cup N^=)$ and $\hat{y}^{(t)}_i = y^{(t)}_i$ for each $t$ and $i\in N\setminus N^>$.

We will do similar operations for indices in $N^<$. 
Let us choose an arbitrary $t$ with $z^{(t)}_{j^*}=\beta_1+1$.
We show that we can choose $j^{(t)}\in N^<$ with $z^{(t)}_{j^{(t)}}\le\beta_1-1$ such that $\hat{y}^{(t)}-\chi_{j^*}+\chi_{j^{(t)}}\in\bar{\B}_C$ 
by applying \Cref{prop:B-EXC-M}.
We denote $N^\geq = N^>\cup N^=$.
Let also 
\begin{align}
\bar{\B}^< = \{ \tilde{y}\in \mathbb{R}^{N^<} \mid \tilde{y}(N^<)=f_C(N)-f_C(N^{\geq}) \text{ and } \tilde{y}(X)\geq f_C(X\cup N^{\geq})-f_C(N^{\geq}) \ (\forall X \subseteq N^<)\}.
\end{align}
(the contraction of $\bar{\B}_C$ by $N^\geq$), and $\ddot{\B}^<$ be the M-convex set induced from $\bar{\B}^<$.
Then it holds that $\hat{y}^{(t)}_{N^<}\in \ddot{\B}^<$, $y'_{N^<}\in \bar{\B}^<$ and $j^* \in \supp^+(\hat{y}^{(t)}_{N^<}-y'_{N^<})$.
We apply \Cref{prop:B-EXC-M} to them. 
Then we can choose an index $j^{(t)}\in \supp^-(\hat{y}^{(t)}_{N^<}-y'_{N^<})$ such that 
$\hat{y}^{(t)}_{N^<}-\chi_{j^*}+\chi_{j^{(t)}}\in\bar{\B}^<$.
Then we can observe that 
$$\hat{y}^{(t)}-\chi_{j^*}+\chi_{j^{(t)}}\in\bar{\B}_C.$$
Indeed, for any $X$ with $j^*\in X$ and $j^{(t)} \notin X$, since $y^{(t)}(X\cap N^<)-1 \geq f_C((X\cap N^<)\cup N^{\geq})-f_C(N^{\geq})$ and $y^{(t)}(X\cap N^\geq)\geq f_C(X\cap N^\geq)$, we have
\begin{align}
    y^{(t)}(X) = y^{(t)}(X\cap N^<) + y^{(t)}(X\cap N^\geq)
    > f_C(X\cap N^<) -f_C(N^\geq)+ f_C(X\cup N^\geq)
    \geq f_C(X).
\end{align}
Moreover, $z^{(t)}_{j^{(t)}} < y'_{j^{(t)}} + x'_{j^{(t)}} = z'_{j^{(t)}} = z_{j^{(t)}} < \beta$, which implies that 
\begin{align}
    \hat{y}^{(t)}_{j^{(t)}} +1 + x'_{j^{(t)}} =  z^{(t)}_{j^{(t)}} + 1 \leq \beta.\label{eq:jt}
\end{align}
For simplicity, let $j^{(t)}=j^*$ for each $t$ with $z^{(t)}_{j^*}=\beta$.
Then, $$y''\coloneqq\sum_{t=1}^r\lambda^{(t)}\cdot (\hat{y}^{(t)}-\chi_{j^*}+\chi_{j^{(t)}})\in\bar{\B}_C.$$

Let $z''=x'+y''$.
Note that this operation to produce $z''$ first reduces the value of elements more than $\beta$ while keeping them at least $\beta$ by \eqref{eq:i^t}, and then increases the value of elements less than $\beta$ while keeping them at most $\beta$ by \eqref{eq:jt}.
In other words, $\beta_1\le z''_i\le z^*_i$ for $i\in N^>$, $z''_i= z^*_i$ for $i\in N^=$, and
$z^*_i\le z''_i\le \beta_1$ for $i\in N^<$.
Therefore, $\Phi(z'')<\Phi(z^*)$ holds.
This contradicts to the optimality of $z^*$.
\end{proof}

Next, we show that $z^*_i$ is at most $\beta_1$ for all $i\in N$.
Recall that $\beta_1=\max\{\lceil f_E(X)/|X|\rceil\mid \emptyset\ne X\subseteq N\}$ by~\eqref{eq:beta_j}.
\begin{lemma}\label{lem:S1-upper}
$z^*_i\le \beta_1$ for all $i\in N$.
\end{lemma}
\begin{proof}
Define the sets $N^>=\{i\in N\mid z^*_i>\beta_1\}$,  $N^==\{i\in N\mid z^*_i=\beta_1\}$, and  $N^<=\{i\in N\mid z^*_i<\beta_1\}$.
By Lemma~\ref{lem:div-layer}, 
$y^*(N^>)=f_C(N^>)$
and $y^*(N^>\cup N^=)=f_C(N^>\cup N^=)$.

Suppose to the contrary that $N^>$ is nonempty.
We construct a graph 
\begin{align}
G=\big(N,\,\big\{(i,j)\in N^2\mid x^*-\chi_i+\chi_j\in\ddot{\B}_M~\text{or}~y^*-\chi_i+\chi_j\in\bar{\B}_C\big\}\big).
\end{align}
We observe that for any $Z$ with $N^>\subseteq Z\subseteq N^>\cup N^=$, it holds that
\begin{align}\label{eq:claim-edge}
z^*(Z) > f_E(Z)
\end{align}
because $f_M(Z)+f_C(Z)=x^*(Z)+y^*(Z)=z^*(Z)>\beta_1\cdot|Z|\ge f_E(Z)=f_M(Z)+f_C(Z)$,
where the last inequality holds by the definition of $\beta_1$.
This implies the following claim.
\def\leftmargini{10pt}
\begin{quote}
\begin{claim}\label{claim:edge}
    For any $Z$ satisfying \eqref{eq:claim-edge}, 
    there exists an edge from some vertex $i\in Z$ to some $j\in N\setminus Z$.
\end{claim}
\begin{proof}
Let us fix $Z$.
Suppose to the contrary that $x^*-\chi_i+\chi_j\not\in\ddot{\B}_M$ and $y^*-\chi_i+\chi_j\not\in\bar{\B}_C$ for all $i\in Z$ and $j\in N\setminus Z$.
In a similar way to Proposition~\ref{prop:B-EXC-XZ}, for each $i\in Z$ and $j\in N\setminus Z$, let $X_{ij}$ be a subset of $N$ such that $x^*(X_{ij})=f_M(X_{ij})$, $i\in X_{ij}$, and $j\not\in X_{ij}$.
Then, since $\bigcup_{i\in Z}\bigcap_{j\in N\setminus Z}X_{ij}$ coincides with $Z$, we have $x^*(Z)=f_M(Z)$ by \Cref{prop:tightsets}.
For each $i\in Z$ and $j\in N\setminus Z$ let $Y_{ij}$ be a subset of $N$ such that $y^*(Y_{ij})\in[f_C(Y_{ij}),f_C(Y_{ij})+1)$, $i\in Y_{ij}$, and $j\not\in Y_{ij}$.
Let $Y'_{ij}\coloneqq (Y_{ij}\cup N^>)\setminus N^<$.
Then, we have
\begin{align}
    y^*(Y_{ij}')
    &= y^*(N^>)+y^*(Y_{ij})-y^*(Y_{ij}\setminus N^=)\\
    &< f_C(N^>)+f_C(Y_{ij})+1-y^*(Y_{ij}\setminus N^=)\\
    &\le f_C(N^>\cup Y_{ij})+f_C(N^>\cap Y_{ij})+1-y^*(Y_{ij}\setminus N^=)\\
    &\le f_C(N^>\cup Y_{ij})-y^*(N^<\cap Y_{ij})+1\\
    &\le f_C(Y_{ij}')+f_C(N^>\cup N^=\cup Y_{ij})-f_C(N^>\cup N^=)-y^*(N^<\cap Y_{ij})+1\\
    &= f_C(Y_{ij}')+f_C(N^>\cup N^=\cup Y_{ij})-y^*(N^>\cup N^=\cup Y_{ij})+1
    \le f_C(Y_{ij}')+1.
\end{align}
Since $y^*_i=\beta_1-x^*_i$ is an integer for any $i\in N^=$ by the choice, $y^*(N^>)=f_C(N^>\cup N^=)-y^*(N^=)$ is also an integer. 
This implies that $y^*(Y'_{ij}) = y^*(N^>) + y^*(Y_{ij}\cap N^=)$ is an integer, and hence we obtain $y^*(Y_{ij}')=f_C(Y'_{ij})$.
Thus, $\bigcup_{i\in Z\cap N^=}\bigcap_{j\in N^=\setminus Z}Y'_{ij}$ coincides with $Z$, and we have $y^*(Z)=f_M(Z)$ by \Cref{prop:tightsets}.
However, this contradicts to \eqref{eq:claim-edge}.
\end{proof}
\end{quote}

\begin{quote}
\begin{claim}\label{claim:find-path}
There exist paths in $G$ from some vertex in $N^>$ to some vertex in $N^<$.
\end{claim}
\begin{proof}
    By setting $Z=N^>$, Claim~\ref{claim:edge} implies that the graph $G$ has an edge from some $i_1\in N^>$ to some $i_2\in N^<\cup N^=$.
    If $i_2 \in N^<$, we have done.
    Let us assume $i_2 \in N^=$.
    Taking $Z=N^>\cup\{i_2\}$, we also use Claim~\ref{claim:edge} to see that $G$ has an edge from some $i_3\in N^>\cup \{i_2\}$ to some $i_4\in (N^<\cup N^=)\setminus \{i_2\}$.
    If $i_4 \in N^<$, then either path $i_3, i_4$ ($i_3\in N^>$) or $i_1, i_2, i_4$ ($i_2=i_3$) is the one what we look for.
    Otherwise, we see that a path from some vertex in $N^>$ ($i_1$ or $i_3$) to $i_4$.
    
    After repeating this procedure $k$ times, we obtain an edge $(i_{2k-1},i_{2k})$ in $G$.
    If $i_{2k}\in N^<$, we obtain a path from some vertex in $N^>$ to $i_{2k}$ via $i_{2k-1}$; otherwise, we obtain a path from some vertex in $N^>$ to $i_{2k}$.
    Since $i_2,i_4,\dots,i_{2k}$ are all distinct, we can find $i_{2k} \in N^<$ after at most $|N^=|+1$ repetitions.
    This completes the proof of the claim.
\end{proof}
\end{quote}

By the above claim and \Cref{lem:transfer} with $N'=N$ and $\beta=\beta_1$, there exists a vector $z$ with $\Phi(z)<\phi(z^*)$, which contradicts to the optimality of $z^*$.
This completes the proof of Lemma~\ref{lem:S1-upper}.
\end{proof}

We then prove that $z^*_i$ is at least $\beta_1-1$ for all $i\in N_1$.
Recall that $N_1$ is the smallest subset of $N$ maximizing $f_E(X)-(\beta_1-1)|X|$.
\begin{lemma}\label{lem:S1-lower}
$z^*_i\ge \beta_1-1$ for all $i\in N_1$.
\end{lemma}
\begin{proof}
Let $N_1^>\coloneqq\{i\in N_1\mid z^*_i(E)> \beta_1-1\}$, $N_1^=\coloneqq\{i\in N_1\mid z^*_i(E)= \beta_1-1\}$, and $N_1^<\coloneqq\{i\in N_1\mid z^*_i< \beta_1-1\}$. 
Suppose to the contrary that $N_1^<$ is nonempty.
We construct a graph 
\begin{align}
G=\big(N_1,\,\big\{(i,j)\in N_1\times N_1\mid x^*-\chi_i+\chi_j\in\ddot{\B}_M~\text{or}~y^*-\chi_i+\chi_j\in\bar{\B}_C\big\}\big).
\end{align}
Recall that $N_1$ is the smallest subset of $N$ maximizing $f_E(S)-(\beta_1-1)|S|$.
For any $X$ with $N_1^<\subseteq X\subseteq N_1^<\cup N_1^=$, we have $f_E(N_1\setminus X)-(\beta_1-1)|N_1\setminus X|\le f_E(N_1)-(\beta_1-1)|N_1|$, which implies $f_E(N_1)-f_E(N_1\setminus X)\ge (\beta_1-1)|X|>z^*(X)$.
In other words, 
$$z^*(Z) = z^*(N_1)-z^*(N_1\setminus Z) \geq f_E(N_1)-z^*(N_1\setminus Z)> f_E(Z)$$
holds for any $Z$ with $N_1^> \subseteq Z \subseteq N_1^>\cup N_1^=$.
By a similar proof of Claim~\ref{claim:edge}, 
we see that there exists an edge from some vertex in $Z$ to $N_1\setminus Z$ for any $Z$ with $N_1^> \subseteq Z \subseteq N_1^>\cup N_1^=$.

Therefore, by a similar way to Claim~\ref{claim:find-path}, we can find a path from some vertex in $N_1^>$ to some in $N_1^<$.
Then, by applying \Cref{lem:transfer} with $N'=N_1$ and $\beta=\beta_1-1$, we can reduce the value of $\Phi$, which is a contradiction.
Hence, $N_1^<$ is empty and $z_i^*\ge \beta_1-1$ for all $i\in N_1$.
\end{proof}

Finally, we show that we cannot decrease the values of elements in $N_1$ (which have high values) anymore. 
In the words of fair allocation, this means that goods not required to be assigned to $N_1$ are not assigned to $N_1$.
\begin{lemma}\label{lem:S1-complete}
$x^*(N_1)=f_M(N_1)$, $y^*(N_1)=f_C(N_1)$, and $z^*(N_1)=f_E(N_1)$.
\end{lemma}
\begin{proof}
It is sufficient to prove that $z^*(N_1)=f_E(N_1)$.
Let $N^>\coloneqq\{i\in N\mid z^*_i> \beta_1-1\}$, $N^=\coloneqq\{i\in N \mid z^*_i= \beta_1-1\}$, and $N^<\coloneqq\{i\in N\mid z^*_i< \beta_1-1\}$. 

Suppose to the contrary that $z^*(N_1)>f_E(N_1)$.
We construct a graph 
\begin{align}
G=\big(N,\,\big\{(i,j)\in N\times N\mid x^*-\chi_i+\chi_j\in\ddot{\B}_M~\text{or}~y^*-\chi_i+\chi_j\in\bar{\B}_C\big\}\big).
\end{align}
Since $N_1 \in \argmax_{S\subseteq N} f_E(S)-(\beta_1-1)|S|$, 
we have $f_E(N_1)-(\beta_1-1)|N_1|\ge f_E(X)-(\beta_1-1)|X|$ for any $X\subseteq N$.
For any $X$ with $N^>\subseteq X\subseteq N^>\cup N^=$, it holds that
\begin{align}
    z^*(X)
    &=z^*(N^>)+z^*(X\setminus N^>)\\
    &= z^*(N^>)+(\beta_1-1)\cdot (|X|-|N^>|)\\
    &\ge z^*(N^>\cap N_1)+(\beta_1-1)\cdot (|X|-|N^>\cap N_1|)\\
    &= z^*(N_1)-(\beta_1-1)\cdot (|N_1|-|N^>\cap N_1|)+(\beta_1-1)\cdot (|X|-|N^>\cap N_1|) &&(\text{by \Cref{lem:S1-lower}})\\
    &= z^*(N_1)-(\beta_1-1)|N_1|+(\beta_1-1)|X|\\
    &> f_E(N_1)-(\beta_1-1)|N_1|+(\beta_1-1)|X| &&(\text{by assumption})\\
    &\ge f_E(X)-(\beta_1-1)|X|+(\beta_1-1)|X|=f_E(X).
\end{align}
Hence, by the same proofs of Claims~\ref{claim:edge} and \ref{claim:find-path}, there exist paths in $G$ from an agent in $N^>$ to an agent in $N^<$.
Then, by applying \Cref{lem:transfer} with $N'=N$ and $\beta=\beta_1-1$, we can decrease the value of $\Phi$, which is a contradiction.
Hence, we obtain $z^*(N_1)=f_E(N_1)$. This implies that $x^*(N_1)=f_M(N_1)$ and $y^*(N_1)=f_C(N_1)$.
\end{proof}

\subsection{Decomposition}\label{sec:proximity-decomp}
We have shown that $\beta_1-1\le z_i^*\le \beta_1$ for all $i\in N_1$ in \Cref{lem:S1-upper,lem:S1-lower}.
We describe that we can derive a similar result for $N_2, \dots, N_q$.

Let $N'_1=N\setminus N_1$.
For a supermodular function $f$, we denote $f^{(1)}\colon 2^{N_1'}\to\mathbb{Z}$ to be the supermodular function obtained from $f$ by contracting $N_1$, i.e., $f^{(1)}(X)=f(X\cup N_1)-f(N_1)$ for each $X\subseteq N_1'$.
We consider the M-convex set $\ddot{\B}_M^{(1)}$ of $f_M^{(1)}$, integral base-polyhedra $\bar{\B}_C^{(1)}$ of $f_C^{(1)}$, and $\B_E^{(1)}=\ddot{\B}_M^{(1)}+\bar{\B}_C^{(1)}$.

By \Cref{lem:S1-complete}, we have $z^*_{N'_1}\in \B_E^{(1)}$.
In addition, for any $z_{N_1'}\in\B_E^{(1)}$, an extended vector $z=(z^*_{N_1},z_{N_1'})$ is contained in $\B_E$ because 
\begin{align}
z(X)
&\ge f_E(X\cap N_1)+f_E^{(1)}(X\cap N_1')
=f_E(X\cap N_1)+f_E(X\cup N_1)-f(N_1)\ge f_E(X)
\end{align}
for all $X\subseteq N$ by the supermodularity of $f_E$.
Hence, we obtain the following lemma.
Let $\Phi'\colon \mathbb{R}^{N_1'}\to\mathbb{R}$ be the symmetric strictly convex function such that $\Phi'(z_{N_1'})=\Phi(z^*_{N_1},z_{N_1'})$.
\begin{lemma}\label{lem:S1-separating}
For any $z_{N_1'} \in \B_E^{(1)}$, 
a vector $z=(z^*_{N_1},z_{N_1'})$ is a $\Phi$-minimizer of $\B_E$ if and only if $z_{N_1'}$ is a $\Phi'$-minimizer of $\B_E^{(1)}$.
\end{lemma}
Therefore, we can apply the results in Section~\ref{sec:proximity-peak} to $N_1'$,  $f_M^{(1)}$, $f_C^{(1)}$ and $f_E^{(1)}$, and repeat the same procedure.
By the definition of the canonical partition and the essential values, we obtain the following lemma.
\begin{lemma}\label{lem:decomposition}
For each $j=1,\dots,q$, it holds that $\beta_j-1\le z_i^*\le \beta_j$ for every $i\in N_j$.
\end{lemma}

Now, we are ready to prove \Cref{thm:proximity}.
\begin{proof}[Proof of \Cref{thm:proximity}]
Recall that the partition $N_1, \dots, N_q$ is the canonical partition of $\ddot{\B}_E$, and $\beta_1,\dots, \beta_q$ are the corresponding essential values.
Let $\hat{N}_1, \dots, \hat{N}_r$ be the principal partition of $\bar{\B}_E$, and let $\lambda_1, \dots, \lambda_r$ be the critical values.
    We fix $i\in N$ arbitrarily.
    Let $j\in [q]$ and $k \in [r]$ be the unique indices such that $i \in N_j$ and $i \in \hat{N}_k$.
    By invoking \Cref{thm:canonical}, we have $\beta_j=\lceil \lambda_k\rceil$.
    Consequently, by \Cref{thm:principal}  and \Cref{lem:decomposition}, we obtain $z_i^*\le \beta_j=\lceil\lambda_k\rceil=\lceil\bar{z}_i\rceil$.

    To show the other inequality, let $\Phi'\colon \mathbb{R}^N\to\mathbb{R}$ be a symmetric strictly convex function such that $\Phi'(z)=\Phi(-z)$ for $z\in\mathbb{R}^N$.
    Let $\B_E'=-\B_E$ and $\bar{\B}_E'=-\bar{\B}_E$.
    Then, $-z^*$ and $-\bar{z}$ are $\Phi'$-minimizers of $\B_E'$ and $\bar{\B}_E'$, respectively.
    By applying the same argument as above, we obtain $-z_i^*\le \lceil -\bar{z}_i\rceil$, which is equivalent to $z_i^*\ge\lfloor\bar{z}_i\rfloor$.
\end{proof}

\subsection{Structures in Fair Allocation}\label{sec:proximity-fair}
Before concluding this section, we discuss additional structures that can be established in the case of fair allocations.
We partition the goods according to the canonical partition as follows.
Let $M_1$ and $C_1$ denote the subset of indivisible goods $M$ and divisible goods $C$, respectively, that must be allocated to agents in $N_1$.
We iteratively define $M_j$ and $C_j$ as the subset of $M\setminus\bigcup_{j'=1}^{j-1}M_{j'}$ and $C\setminus\bigcup_{j'=1}^{j-1}C_{j'}$, respectively, that must be allocated to agents in $\bigcup_{j=1}^i N_j$.
In other words, $M_j$ and $C_j$ ($j=1,\dots, q$) is defined as
\begin{align}
    M_j &= \textstyle\left\{g\in M\setminus\bigcup_{j'=1}^{j-1}M_{j'} \mid v_{ig}=0~(\forall i\in N\setminus \bigcup_{j'=1}^j N_{j'}) \right\}, \label{eq:canonical indivisible}\\
    C_j &= \textstyle\left\{c\in C\setminus\bigcup_{j'=1}^{j-1}C_{j'} \mid v_{ic}=0~(\forall i\in N\setminus \bigcup_{j'=1}^j N_{j'})\right\}. \label{eq:canonical divisible}
\end{align}
We refer $M_1,\dots,M_q$ and $C_1,\dots,C_q$ as the canonical partitions of the indivisible goods and the divisible goods, respectively.

\begin{theorem}\label{thm:structure fair}
    For any allocation $\pi^*$ whose utility vector is a $\Phi$-minimizer over $\B_E$, it holds that $\sum_{i\in N_j}\pi^*_{ie}=1$ for every good $e\in M_j\cup C_j$ and $j=1,2,\dots,q$.
\end{theorem}
\begin{proof}
    Let $z^*$ be the utility vector of $\pi^*$.
    Let $x^*$ and $y^*$ be the utility vectors for indivisible and divisible goods in $\pi^*$, respectively.
    Thus, $z^*=x^*+y^*$.
    
First, since $f_M(N_1)=|M_1|$ and $f_C(N_1)=|C_1|$, we have $z^*(N_1)=|M_1 \cup C_1|$ by \Cref{lem:S1-complete}.
Next, let $N_j'=N\setminus \bigcup_{j'=1}^j N_{j'}$ for $j=1,\dots, q-1$.
By \Cref{lem:S1-separating}, $z^*_{N_1'}$ is a $\Phi'$-minimizer, where $\Phi'(z)=\phi(z^*_{N_1}, z)$.
Thus by the definition of $N_2$ and \Cref{lem:S1-complete} again, $x^*_{N_1'}(N_2) = f^{(1)}_M(N_2) = f_M(N_1\cup N_2)-f_M(N_1) = |M_2|$ and $y^*_{N_1'}(N_2) = f^{(1)}_C(N_2) = |C_2|$.
By iteratvely applying this argument, we observe that $z^*(N_{j}) = |M_j \cup C_{j}|$ for $j=2,\dots, q$.

Because agents in $N_q$ want only the goods in $M_q \cup C_q$, these goods are allocated to the agents in $N_q$.
Then, for each $j=q-1, \dots, 1$, since agents in $N_{j}$ want only the goods in $\bigcup_{j'=j}^q (M_{j'} \cup C_{j'})$ but the goods in $\bigcup_{j'=j+1}^q (M_{j'} \cup C_{j'})$ are allocated to agents in $N_{j+1}\cup \dots \cup N_q$, the goods in $M_{j} \cup C_{j}$ are allocated to agents in $N_j$.
Therefore, the theorem holds.
\end{proof}

\section{Tractability for Identical Divisible Goods}\label{sec:identical-divisible}

In this section, we focus on the setting where all the divisible goods are identical, i.e., $v_{ic}=v_{ic'}$ for any $c,c'\in C$ and $i\in N$.
Assume that $C$ is non-empty, and let $\Phi$ be a symmetric strictly convex function.
The main result of this section is a polynomial-time algorithm to find a $\Phi$-fair allocation (i.e., an allocation whose utility vector is a $\Phi$-minimizer).
This implies that we can find an allocation that maximizes the Nash welfare for mixed goods in polynomial time when the divisible goods are identical.

Our algorithm utilizes the structures discussed in the previous section.
Namely, we seek allocations whose utility vectors satisfy the statements in \Cref{lem:decomposition} and \Cref{thm:structure fair}, which are necessary conditions to be optimal.
Then, we find the best allocation from them.
In the following, we describe our algorithm, which is summarized in \Cref{alg:identical_div}.

The algorithm first finds a discrete $\Phi$-minimizer $\ddot{z}^*\in\argmin_{z\in\ddot{\B}_E}\Phi(z)$, the canonical partition $N_1,\dots,N_q$, and the essential values $\beta_1,\dots,\beta_q$ of $\ddot{\B}_E$. 
These can be computed in polynomial time by \Cref{thm:canonical}.
Note that $\ddot{z}^*$ is an optimal utility vector if every good is assumed to be indivisible.
In addition, the canonical partition of the indivisible goods $M_1,\dots,M_q$ and that of the divisible ones $C_1,\dots,C_q$ can be calculated easily by~\eqref{eq:canonical indivisible} and \eqref{eq:canonical divisible}.
Since all divisible goods are identical, only one of $C_1,\dots,C_q$ is non-empty.
Let $j^*$ be the index such that $C_{j^*}=C$.

Next, we decide on an allocation for agents in $N_j$ for each $j \neq j^*$.
Let us fix $j \neq j^*$.
\Cref{thm:structure fair} implies that in an optimal allocation, the agents in $N_j$ receive only the indivisible goods in $M_j$.
Moreover, by \Cref{lem:decomposition}, some agents in $N_j$ must receive $\beta_j$ goods and the others must receive $\beta_j-1$.
Because $\beta_j-1 \leq \ddot{z}^*_i \leq \beta_j$ holds by \Cref{thm:canonical}, we allocate goods in $M_j$ so that each agent $i\in N\setminus N_{j^*}$ receives $\ddot{z}^*_i$ goods. 
Such an allocation can be computed by solving a bipartite matching problem.\footnote{It can also be calculated directly with a method of Harvey et al.~\cite{HLLT2006}.}
For agents in $N_j$, the utility vector of this allocation is value-equivalent to an optimal one.
Thus, we have found an optimal allocation for agents in $N_j$.

The remaining task is to determine the allocation of $M_{j^*}\cup C$ to agents $N_{j^*}$.
Since the optimal allocation of $M_{j^*}\cup C$ depends on $\Phi$ as we saw in \Cref{ex:decmin=SOS,ex:incmax=SOS}, we conduct an enumeration-based approach rather than performing a full characterization.

Let $\pi^*$ be an optimal allocation.
Let $N^+_{j^*}$ be the set of agents in $N_{j^*}$ who desire the divisible goods, i.e., $N^+_{j^*}=\{i\in N_{j^*} \mid v_i(c)=1~(\forall c \in C)\}$. Let $N^-_{j^*}=N_{j^*}\setminus N^+_{j^*}$.
The following lemma indicates that there are a finite number of candidates for a $\Phi$-minimizer.
\begin{lemma}\label{lem:cake-equal}
All the agents who receive divisible goods (i.e., $\pi^*_i(C)>0$) have the same utility.
\end{lemma}
\begin{proof}
     Assume that there exist two agents $i,i' \in N_{j^*}^+$ with $\pi^*_i(E) > \pi^*_{i'}(E)$, $\pi^*_i(C)>0$, and $\pi^*_{i'}(C)>0$.
     Then by equation~\eqref{eq:ssc}, we can decrease the value of $\Phi$ by transferring a fraction $\varepsilon \in (0, \pi^*_i(E) - \pi^*_{i'}(E))$ of a divisible good from agent $i$ to agent $i'$.
     This contradicts that the utility vector of $\pi^*$ is a $\Phi$-minimizer.
\end{proof}

Let $k$ be the number of agents in $N^+_{j^*}$ who receive $\beta_{j^*}$ indivisible goods and let $\ell$ be the total number of indivisible goods received by agents in $N^+_{j^*}$.
Note that $k < |N^+_{j^*}|$.
The key observation is the following lemma.
\begin{lemma}\label{lem:optimal alloc}
    The following properties hold:
    \begin{enumerate}
        \item $|N^+_{j^*}|\cdot (\beta_{j^*}-1)+k\le \ell+|C|\le |N^+_{j^*}|\cdot \beta_{j^*}$;
        \item there exist $X\subseteq N^+_{j^*}$ such that $|X|=k$ and 
        \begin{enumerate}
            \item for each $i\in X$:  $\pi^*_i(M_{j^*})=\beta_{j^*}$ and $\pi^*_i(C)=0$;
            \item 
            for each $i\in N^+_{j^*}\setminus X$: $\pi^*_i(M_{j^*})\le \beta_{j^*}-1$ and $\pi^*_i(E)=\beta_{j^*}-(|N^+_{j^*}|\cdot \beta_{j^*}-\ell-|C|)/(|N^+_{j^*}|-k)$;
        \end{enumerate}
        \item $\pi^*_i(M_{j^*})\in \{\beta_{j^*},\beta_{j^*}-1\}$ and $\pi^*_i(C)=0$ for each $i\in N^-_{j^*}$.
    \end{enumerate}
\end{lemma}
\begin{proof}
Let $X= \{i\in N^+_{j^*}\mid \pi^*_i(M_{j^*})=\beta_{j^*}\}$.
By \Cref{lem:decomposition}, it holds that $\beta_{j^*}-1 \leq \pi^*_i(E) \leq \beta_{j^*}$ for each $i \in N_{j^*}$.
Hence, it is easy to see properties 2(a) and 3.
In addition, by the definitions of $k$ and $\ell$, we have $|N^+_{j^*}|\cdot (\beta_{j^*}-1)+k\leq \sum_{i \in N_{j^*}^+} \pi^*_i (E) = \ell+|C|$, and $\sum_{i \in N_{j^*}^+} \pi^*_i (E) \leq |N^+_{j^*}|\cdot \beta_{j^*}$. 

It remains to see property 2(b).
By definition of $X$, we have $\pi^*_i(M_{j^*})\le \beta_{j^*}-1$ for each $i \in N_{j^*}^+ \setminus X$.
Let $Y$ be the set of agents in $N_{j^*}^+ \setminus X$ who receive divisible goods.
By \Cref{lem:cake-equal}, $\pi^*_i(E)=\pi^*_{i'}(E)$ for any $i, i' \in Y$. 
If $Y \neq N_{j^*}^+ \setminus X$ and $\pi^*_i(E) > \beta_{j^*}-1$ for some $i\in Y$, then giving some portion of a divisible good from agent $i$ to agent in $N_{j^*}^+ \setminus (X\cup Y)$ decreases the value of $\Phi$ by \eqref{eq:ssc}, which is a contradiction.
Thus, $Y=N^+_{j^*}\setminus X$ or $\pi^*_i(E) = \beta_{j^*}-1$ for all $i\in Y$.
Hence, in either case, all the agents in $N_{j^*}^+ \setminus X$ have the same utility.
Therefore, it holds that
\begin{align}
\pi^*_i(E)=\frac{\ell+|C|-k\cdot \beta_{j^*}}{|N^+_{j^*}|-k}=\beta_{j^*}-\frac{|N^+_{j^*}|\cdot \beta_{j^*}-\ell-|C|}{|N^+_{j^*}|-k} \quad (\forall i \in N_{j^*}^+ \setminus X),
\end{align}
and property 2(b) is satisfied.
\end{proof}
Since \Cref{lem:optimal alloc} specifies the utility vector (i.e., $\pi^*_i(E)$ for each $i\in N_{j^*}$) of an optimal allocation up to value-equivalence, it suffices to find an allocation whose utility vector satisfies the statement in \Cref{lem:optimal alloc}.
In fact, if we are given $\ell$, an optimal allocation can be computed as follows.
For each $k=0,\dots, |N_{j^*}^+|$ such that property 1 in \Cref{lem:optimal alloc} is satisfied, 
\begin{enumerate}
\item find an allocation $\pi^{k,\ell} \in \{0,1\}^{N_{j^*}\times M_{j^*}}$ of indivisible goods in $M_{j^*}$ such that 
(a) $|\{i\in N^+_{j^*}\mid \pi_i(M_{j^*})=\beta_{j^*}\}| \leq k$,
(b) $\pi_i(M_{j^*})\le \beta_{j^*}$ for each $i\in N^+_{j^*}$,
(c) $\sum_{i\in N^+_{j^*}}\pi_i(M_{j^*})=\ell$,
(d) $\pi_i(M_{j^*})\in \{\beta_{j^*},\beta_{j^*}-1\}$ for each $i\in N^-_{j^*}$;

\item if $\pi^{k,\ell}$ exists, 
let $\hat{\pi}^{k,\ell}$ be the allocation by allocating indivisible goods according to $\pi^{k, \ell}$, and allocating divisible goods by a water-filling policy so that 
\begin{align}
\hat{\pi}^{k,\ell}_i(c)=
\frac{1}{|C|}\cdot\left(\beta_{j^*}-\frac{|N^+_{j^*}|\cdot\beta_{j^*}-\ell-|C|}{|N^+_{j^*}|-k}-\pi^{k,\ell}_i(M)\right)
\end{align}
for each $i\in N^+_{j^*}$ such that $\pi_i^{k,\ell}(M)<\beta_{j^*}$ and $c\in C$.
\end{enumerate}
Let us see that this indeed works.
Since $\Phi(\hat{\pi}^{k,\ell})\leq \Phi(\hat{\pi}^{k+1,\ell})$ holds (if $\pi^{k,\ell}$ exists), 
the smallest $k$ such that $\pi^{k,\ell}$ exists is the number in \Cref{lem:optimal alloc}.
For such an integer $k$, we have $|\{i\in N^+_{j^*}\mid \pi_i(M_{j^*})=\beta_{j^*}\}| = k$, and the properties in \Cref{lem:optimal alloc} are satisfied except the allocation of divisible goods.
Once we have decided on an allocation of indivisible goods, an optimal allocation of divisible goods is found by the water-filling policy.

Now, we explain how to find an allocation $\pi^{k, \ell}$ at step 1 in polynomial time.
We reduce this problem to the submodular flow problem, which can be solved in polynomial time~(see, e.g., Murota~\cite{Murota:DCA}).
Let $G=(V,A)$ be a directed graph constructed as follows.
The set of vertices $V$ is $M_{j^*}\cup N_{j^*}\cup N'_{j^*}$ where $N'_{j^*}$ is a set of copy $i'$ of each $i\in N_{j^*}$.
The set of edges $A$ is $A_1\cup A_2 \cup A_3$ where $A_1=\{(g,i')\in M_{j^*}\times N'_{j^*}\mid v_i(g)=1\}$, $A_2=\{(i',i)\mid i\in N^+_{j^*}\}$, and $A_3=\{(i',i)\mid i\in N^-_{j^*}\}$.
We define $\underbar{c}, \bar{c} \colon A \to \mathbb{Z}$ as $\underbar{c}(a)=0$ and $\bar{c}(a)=1$ for each $a\in A_1$; $\underbar{c}(a)=0$ and $\bar{c}(a)=\beta_{j^*}$ for each $a \in A_2$; $\underbar{c}(a)=\beta_{j^*}-1$ and $\bar{c}(a)=\beta_{j^*}$ for each $a \in A_3$.
In addition, let $f_{k, \ell}\colon 2^V\to\mathbb{Z}$ be a function such that 
\begin{align}
f_{k, \ell}(X)=\varphi_{k, \ell}(|X\cap N^+_{j^*}|)+(|M_{j^*}|-\ell)\bm{1}_{X\cap N^-_{j^*}\ne \emptyset}-|X\cap M_{j^*}| \quad (\forall X\subseteq V),
\end{align}
where $\varphi_{k, \ell}(h)=\min\{\beta_{j^*}h,\, (\beta_{j^*}-1)h+k,\,\ell\}$, and $\bm{1}_{X\cap N^-_{j^*}\ne\emptyset}$ takes the value $1$ if $X\cap N^-_{j^*}\ne\emptyset$ and $0$ otherwise.
We remark that $f_{k, \ell}$ is a submodular function, and $f_{k,\ell}(V)=0$ since $\ell \leq \sum_{i\in N_{j^*}^+} \pi^*_i(M_{j^*})\leq (\beta_{j^*}-1)|N_{j^*}^+|+k\le \beta_{j^*}|N_{j^*}^+|$.

\begin{lemma}\label{lem:submodular flow}
    There exists an allocation $\pi\in\{0,1\}^{N_{j^*}\times M_{j^*}}$ satisfying (a)--(d) if and only if there exists an integral flow $\xi\colon A\to\mathbb{Z}$ satisfying $\underbar{c}(a)\le \xi(a)\le \bar{c}(a)$ (capacity constraints) and a constraint (called supply specification) that the boundary $\partial\xi\in\mathbb{Z}^V$ of the flow $\xi$, which is defined by
\begin{align}
    \partial\xi(v)=\sum_{a=(v,u)\in A}\xi(a)-\sum_{a=(u,v)\in A}\xi(a),
\end{align}
is in the M-convex set $\ddot{\B}=\{x\in\mathbb{Z}^V\mid x(V)=f_{k,\ell}(V)\text{ and }x(X)\le f_{k, \ell}(X)~(\forall X\subseteq V)\}$. 
\end{lemma}
\begin{proof}
Let $\pi$ be an allocation satisfying (a)--(d).
Let $\xi$ be an integral flow such that $\xi(g,i')=\pi_{i}(g)$ for each $i \in N_{j^*}, g\in M_{j^*}$, and $\xi(i',i)=\sum_{g} \xi(g,i')$ for each $i \in N_{j^*}$. 
Since $\xi(i',i)=\pi_i(M_{j^*})$, the conditions (b) and (d) ensure the capacity constraints.
By construction, $\partial \xi(g)=-1$ for each $g\in M_{j^*}$, $\partial \xi(i')=0$ for each $i'\in N'_{j^*}$, and $\partial \xi(i)=\pi_i(M_{j^*})$ for each $i\in N_{j^*}$.
To see that $\partial \xi \in \ddot{B}$, fix $X\subseteq V$.
We have $\partial \xi(X)=\sum_{i\in X\cap N_{j^*}^+} \pi_i(M_{j^*})+\sum_{i\in X\cap N_{j^*}^-} \pi_i(M_{j^*})-|X\cap M_{j^*}|$.
Here, we have $\sum_{i\in X\cap N_{j^*}^+} \pi_i(M_{j^*}) \leq \beta_{j^*} |X\cap N_{j^*}^+|$ by condition (b),  $\sum_{i\in X\cap N_{j^*}^+} \pi_i(M_{j^*}) \leq (\beta_{j^*}-1)|X\cap N_{j^*}^+|+k$ by conditions (a) and (b), and $\sum_{i\in X\cap N_{j^*}^+} \pi_i(M_{j^*}) \leq \ell$ by condition (c).
Moreover, $\sum_{i\in X\cap N_{j^*}^-} \pi_i(M_{j^*}) \leq (|M_{j^*}|-\ell)\bm{1}_{X\cap N^-_{j^*}\ne \emptyset}$ by condition (c).
Thus, we have $\partial \xi(X) \leq f_{k,\ell}(X)$.
In addition, $\partial \xi(V)=0$, and hence $\partial \xi \in \ddot{B}$.

Conversely, let $\xi$ be a feasible integral flow.
Define $\pi$ by $\pi_i(g)=\xi(g,i')$ for each $i\in N_{j^*}$ and $g\in M_{j^*}$.
For each $i\in N_{j^*}$ and $g\in M_{j^*}$, we have $\xi(g,i') \in \{0,1\}$ by the capacity constraint, and $\partial \xi(g) \leq -1$ by the supply specification with $X=\{g\}$.
In addition, we have $\partial \xi(M_{j^*}) \leq f_{k,\ell}(M_{j^*})= -|M_{j^*}|$, $\partial \xi(N_{j^*}\cup N_{j^*}') \leq |M_{j^*}|$ and $\partial \xi(V)=0$, we have $\partial \xi(M_{j^*}) = -|M_{j^*}|$.
Thus, $\sum_{i} \pi_i(g)=1$ for each $g\in M_{j^*}$, and hence $\pi$ is an allocation of $M_{j^*}$.

We observe that $\partial \xi(i')=\sum_{g\colon v_i(g)=1} \xi(g,i')-\xi(i',i) \leq f_{k,\ell}(\{i'\})=0$ for each $i'\in N'_{j^*}$.
Moreover, since $\partial \xi(M_{j^*}\cup N_{j^*}) \leq 0$ and $\partial \xi(V)=0$, we have $\partial \xi(i')=0$ for each $i'\in N'_{j^*}$, which implies that $\pi_i(M_{j^*})=\sum_{g\colon v_i(g)=1} \xi(g,i') = \xi(i',i)$.
Then we can see that each condition is satisfied.
\begin{description}
    \item[(d)] The capacity constraint for $A_3$ implies that $\pi_i(M_{j^*}) = \xi(i',i) \in \{\beta_{j^*}-1, \beta_{j^*}\}$.

    \item[(c)] We have $\sum_{i\in N_{j^*}^+} \pi_i(M_{j^*})=\partial \xi(N_{j^*}^+) \leq f_{k,\ell}(N_{j^*}^+)=\ell$. Moreover, $\partial \xi(N_{j^*}^+) = -\partial \xi(V\setminus N_{j^*}^+) \geq \ell$ since $\partial \xi(V\setminus N_{j^*}^+) \leq -\ell$.

    \item[(b)] The capacity constraint for $A_2$ implies that $\pi_i(M_{j^*}) = \xi(i',i) \leq \beta_{j^*}$ for each $i\in N_{j^*}^+$.

    \item[(a)] It holds that $\sum_{i\in N_{j^*}^+} \pi_i(M_{j^*})=\partial \xi(N_{j^*}^+) \leq \phi_{k,\ell}(N_{j^*}^+) \leq (\beta_{j^*}-1)|N_{j^*}^+|+k = (\beta_{j^*}-1)(|N_{j^*}^+|-k)+k\beta_{j^*}$. By (b), at most $k$ agents have the utility $\beta_{j^*}$.
\end{description}
This completes the proof.
\end{proof}

It is known that the feasibility of the submodular flow problem can be determined in polynomial time~\cite{fujishige2005}.
Hence, by \Cref{lem:submodular flow}, the existence of an allocation satisfying conditions (a)--(d) can be determined in polynomial time. Moreover, if such an allocation exists, we can find one of such allocations in polynomial time.

Finally, because we do not know $\ell$ in advance, we enumerate all possibilities.
That is, find a best allocation $\pi^{k, \ell}$ for each $\ell=0,1,\dots,|M_{j^*}|$ by the above procedure, and choose the best one.
Then the resulting allocation is as good as an optimal allocation $\pi^*$.

\begin{algorithm}[htb]
    \caption{Allocation algorithm when the divisible goods are identical}\label{alg:identical_div}
    \SetKwInOut{Input}{input}
    \SetKwInOut{Output}{output}
    \Input{A fair allocation instance $(N,M,C,v)$ and a symmetric strictly convex function $\Phi$}
    \Output{A $\Phi$-fair allocation}
    Compute the canonical partition $N_1,\dots,N_q$, the essential values $\beta_1,\dots,\beta_q$, the canonical partition of the indivisible goods $M_1,\dots,M_q$, and the canonical partition of the divisible goods $C_1,\dots,C_q$\;
    Let $j^*$ be the index such that $C_{j^*}=C$\;
    \For{$j\ot 1,\dots,j^*-1,j^*+1,\dots,q$}{
        Allocate $M_j$ to $N_j$ so that each agent receives $\beta_j$ or $\beta_j-1$\;\label{line:assign_Mj}
    }
    Let $N^+_{j^*}\ot\{j\in N_{j^*}\mid v_i(c)=1~(\forall c\in C)\}$ and $N^-_{j^*}\ot\{j\in N_{j^*}\mid v_i(c)=0~(\forall c\in C)\}$\;
    Let $\Pi\ot\emptyset$ be a set of candidate allocations\;
    \For{$k\ot 0,1,\dots,|N^+_{j^*}|$ and $\ell\ot 0,1,\dots,|M_{j^*}|$}{
        \If{
        $|N^+_{j^*}|\cdot (\beta_{j^*}-1)+k\le \ell+|C|\le |N^+_{j^*}|\cdot \beta_{j^*}$
        }{
            Determine the existence an allocation $\pi^{k,\ell} \in\{0,1\}^{N_{j^*}\times M_{j^*}}$ satisfying the following conditions via the submodular flow problem:
            $\{i\in N^+_{j^*}\mid \pi_i(M_{j^*})=\beta_{j^*}\}\le k$,
            $\pi_i(M_{j^*})\le \beta_{j^*}$ for each $i\in N^+_{j^*}$,
            $\sum_{i\in N^+_{j^*}}\pi_i(M_{j^*})=\ell$,
            $\pi_i(M_{j^*})\in \{\beta_{j^*},\beta_{j^*}-1\}$ for each $i\in N^-_{j^*}$\;
            \If{Such an allocation $\pi^{k,\ell}$ exists}{
                Let $\pi$ be an allocation such that indivisible goods are allocated according to \cref{line:assign_Mj} and $\pi^{k, \ell}$, and the divisible goods are allocated to agents in $N_{j^*}^+$ by a water-filling policy\;
                $\Pi\ot \Pi\cup\{\pi\}$\;
            }
        }
    }
    \Return $\pi^*\in \argmin_{\pi\in\Pi}\Phi(\pi(E))$\;
\end{algorithm}

By summarizing the discussions so far, we obtain the main result in this section.
\begin{theorem}[restatement of \Cref{thm:intro-alg}]\label{thm:identical cake}
    Let $\Phi$ be a symmetric strictly convex function.
    If all the divisible goods are identical, \Cref{alg:identical_div} finds a $\Phi$-fair allocation 
    in polynomial time.
\end{theorem}

\section{Hardness for Identical Indivisible Goods}\label{sec:identical-indivisible}
In this section, we show a hardness result when divisible goods are non-identical but indivisible goods are identical.

We prove the NP-hardness of finding a $\Phi$-fair allocation by using the 3-dimensional matching (3DM) problem, which is known to be NP-hard~\cite{GJ1979}.
In the 3DM problem, we are given three sets of elements, $X=\{x_1,\dots,x_n\}$, $Y=\{y_1,\dots,y_n\}$, and $Z=\{z_1,\dots,z_n\}$. We are also given a set of hyperedges $T=\{t_1,\dots,t_m\}$ where $t\in T$ is an ordered triplets in $X\times Y\times Z$. The goal of the problem is to determine if there exists a subset $T'$ of $T$ such that each element from $X$, $Y$, and $Z$ appears exactly once in $T'$.
In other words, our task is to find a perfect matching that covers all the elements in $X$, $Y$, and $Z$ without any repetitions.
\begin{theorem}[restatement of \Cref{thm:intro-NP-hard}]\label{thm:hardness}
For any fixed symmetric strictly convex function $\Phi$,
finding a $\Phi$-minimizer is NP-hard even when indivisible goods are identical.
Hence, finding a $\Phi$-fair allocation is NP-hard.
\end{theorem}
\begin{proof}
Let us fix a symmetric strictly convex function $\Phi$.
Given an instance of 3DM $(X,Y,Z;T)$, we construct an instance of fair allocation $(N,M,C,v)$.
Without loss of generality, we assume $m\ge n$, since otherwise, the 3DM instance is a no-instance.
We construct two agents $s$ and $s'$ for each $s\in X\cup Y\cup Z$ and five agents 
$t,t^{(1)},t^{(2)},t^{(3)},t^{(4)}$ for each $t\in T$. 
Thus, we have the set $N$ of $6n+5m$ agents: $N=\bigcup_{s\in X\cup Y\cup X}\{s,s'\}\cup\bigcup_{t\in T}\{t,t^{(1)},t^{(2)},t^{(3)},t^{(4)}\}$.
We create $n$ indivisible goods $M=\{g_1,g_2,\dots,g_n\}$, which correspond to the hyperedges chosen as a matching. 
We set the valuations so that only the agents in $T$ desire the indivisible goods, i.e., $v_t(g)=1$ ($t\in T$) and $v_i(g)=0$ ($i \in N\setminus T$) for every $g \in M$.
We construct two types of divisible goods.
First, for each $s\in X\cup Y\cup Z$, we create a good $c_s$, which is desired by only $s$ and $s'$.
Second, for each hyperedge $t=(x,y,z) \in T$, we create goods $c_t^{(1)},c_t^{(2)},c_t^{(3)}$, which are desired by only $8$ agents $t,t^{(1)},t^{2},t^{(3)},t^{(4)},x,y,z$.
Thus, we have the set of $3n+3m$ divisible goods, denoted by $C=\bigcup_{s\in X\cup Y\cup Z}\{c_s\}\cup\bigcup_{t\in T}\{c_t^{(1)},c_t^{(2)},c_t^{(3)}\}$.
We summarize the valuation of each agent for each good in Table~\ref{tbl:reduction}.

\begin{table}[htb]
    \centering
    \caption{Agents' valuations in the proof of Theorem~\ref{thm:hardness}.}\label{tbl:reduction}
    \begin{tabular}{c|ccc}
      \toprule
      agents                   & $g_i~(i\in[n])$ & $c_{\hat{s}}~(\hat{s}\in X\cup Y\cup Z)$ & $c_{\hat{t}}^{(1)},c_{\hat{t}}^{(2)},c_{\hat{t}}^{(3)}~(\hat{t}\in T)$\\\midrule
      $s~(s\in X\cup Y\cup Z)$ & $0$   & $1$ if $s=\hat{s}$ and $0$ otherwise  & $1$ if $s\in\hat{t}$ and $0$ otherwise \\
      $s'~(s\in X\cup Y\cup Z)$& $0$   & $1$ if $s=\hat{s}$ and $0$ otherwise  & $0$\\
      $t~(t\in T)$             & $1$   & $0$   & $1$ if $t=\hat{t}$ and $0$ otherwise\\      
      $t^{(1)},t^{(2)},t^{(3)},t^{(4)}~(t\in T)$ & $0$   & $0$   & $1$ if $t=\hat{t}$ and $0$ otherwise\\ 
      \bottomrule
    \end{tabular}
\end{table}

We can observe that there always exists an integral allocation such that the utility of each agent is at most $1$ as follows:
allocate indivisible good $g_i\in M$ to agent $t_i\in N$ for each $i\in[n]$,
divisible good $c_s\in C$ to agent $s\in N$ for each $s\in X\cup Y\cup Z$,
divisible good $c_t^{(j)}\in C$ to agent $t^{(j)}\in N$ for each $t\in T$ and $j\in\{1,2,3\}$.
Thus, in any dec-min \emph{relaxed} allocation, the largest utility is at most $1$.
Hence, the utility of each agent is at most $1$ in any $\Phi$-fair relaxed allocation since $\Phi$-fair and dec-min allocations coincide in relaxed allocations by \Cref{thm:divisible}.
By the proximity (Theorem~\ref{thm:proximity}), the utility of each agent is also at most $1$ in any $\Phi$-fair allocation.

We demonstrate that the 3DM instance is a yes-instance if and only if there exists an allocation such that only $n$ agents have utility $1$ and the other agents have utility $0.6$.
Note that, if such an allocation exists, it must be a $\Phi$-fair allocation.
This is because at least $n$ agents have utility $1$ as there are $n$ indivisible goods.

Suppose that the 3DM instance is a yes-instance.
Let $T'$ be a subset of $T$ such that each element in $X\cup Y\cup Z$ appears exactly once.
We construct a desired allocation $\pi^*$ as follows.
For the indivisible goods $M$, we distribute one to each agent in $T'$ 
(i.e., $\pi^*_{t_{\tau(i)}}(g_i)=1$ for each $i\in [n]$, where $T'=\{t_{\tau(1)},\dots,t_{\tau(n)}\}$).
For each divisible good $c_s$ with $s\in X\cup Y\cup Z$, we allocate $0.4$ fraction of $c_s$ to agent $s\in N$ and $0.6$ fraction to agent $s'\in N$ (i.e., $\pi^*_s(c_s)=0.4$ and $\pi^*_{s'}(c_s)=0.6$).
For each $t=(x,y,z)\in T'$ and $j\in\{1,2,3\}$, we allocate $0.2/3$ fraction of $c_t^{(j)}$ to each of agents $x,y,z$ (i.e., $\pi^*_{x}(c_t^{(j)})=\pi^*_{y}(c_t^{(j)})=\pi^*_{z}(c_t^{(j)})=0.2/3$) and $0.2$ fraction to each of agents $t^{(1)},t^{(2)},t^{(3)},t^{(4)}$ (i.e., $\pi^*_{t^{(1)}}(c_t^{(j)})=\pi^*_{t^{(2)}}(c_t^{(j)})=\pi^*_{t^{(3)}}(c_t^{(j)})=\pi^*_{t^{(4)}}(c_t^{(j)})=0.2$).
For each $t\in T\setminus T'$ and $j\in\{1,2,3\}$, we allocate $0.2$ of $c_t^{(j)}$ to each of agents $t$, $t^{(1)}$, $t^{(2)}$, $t^{(3)}$, and $t^{(4)}$ (i.e., $\pi^*_{t}(c_t^{(j)})=\pi^*_{t^{(1)}}(c_t^{(j)})=\pi^*_{t^{(2)}}(c_t^{(j)})=\pi^*_{t^{(3)}}(c_t^{(j)})=\pi^*_{t^{(4)}}(c_t^{(j)})=0.2$).
Then, we can see that $\pi_i^*(E)$ is $1$ if $i\in T'$ and $0.6$ otherwise.

Conversely, suppose that there exists an allocation $\pi^\dagger$ such that $\pi_i^\dagger(E)$ is $1$ for $n$ agents of $i\in N$ and $0.6$ for the other agents.
Let $N^\dagger$ be the set of agents $i\in N$ with $\pi_i^\dagger(E)=1$.
As there are $n$ indivisible goods $M$, each agent $i\in N^\dagger$ must receive one indivisible good, and hence $N^\dagger\subseteq T$.
Thus, we identify $N^\dagger$ with a set of hyperedges.
We prove that $N^\dagger$ is a perfect matching for the 3DM instance.
Suppose to the contrary that $N^\dagger$ is not a perfect matching in the 3DM instance.
Then, there exist two distinct agents (corresponding to hyperedges) 
$t=(x,y,z)$ and $\hat{t}=(\hat{x},\hat{y},\hat{z})$
in $N^\dagger$ such that they have an intersection.
Then, 
we have $k\coloneqq|\{x,y,z,\hat{x},\hat{y},\hat{z}\}|\in\{3,4,5\}$.
Thus, the $k+6$ divisible goods $\bigcup_{s\in\{x,y,z,\hat{x},\hat{y},\hat{z}\}}\{c_s\}\cup\{c_t^{(1)},c_t^{(2)},c_t^{(3)}\}\cup\{c_{\hat{t}}^{(1)},c_{\hat{t}}^{(2)},c_{\hat{t}}^{(3)}\}$ must be allocated to $2k+8$ agents $\bigcup_{s\in\{x,y,z,\hat{x},\hat{y},\hat{z}\}}\{s,s'\}\cup\{t^{(1)},t^{(2)},t^{(3)},t^{(4)}\}\cup\{\hat{t}^{(1)},\hat{t}^{(2)},\hat{t}^{(3)},\hat{t}^{(4)}\}$.
However, we have $(k+6)/(2k+8)>0.6$ since $k\in\{3,4,5\}$.
Therefore, at least one of the $2k+8$ agents must have a utility greater than $0.6$ in the allocation $\pi^\dagger$, which is a contradiction.
\end{proof}

By the proof of this theorem, we can also prove the following.
\begin{corollary}
The problems of finding an MNW allocation and an optimal egalitarian allocation are both NP-hard, even when indivisible goods are identical.
\end{corollary}

\section*{Acknowledgments}
We are grateful to Kazuo Murota and Warut Suksompong for their helpful comments.
The first author is supported by JSPS KAKENHI Grant Number JP20K19739, JST PRESTO Grant Number JPMJPR2122, and Value Exchange Engineering, a joint research project between R4D, Mercari, Inc.\ and the RIISE. 
The third author is supported by JSPS KAKENHI Grant Numbers JP21K17708 and JP21H03397, Japan.

\bibliographystyle{abbrv}
\bibliography{ref}

\begin{thebibliography}{10}

\bibitem{Anari2018}
N.~Anari, T.~Mai, S.~O. Gharan, and V.~V. Vazirani.
\newblock {Nash} social welfare for indivisible items under separable,
  piecewise-linear concave utilities.
\newblock In {\em Proceedings of the 29th Annual ACM-SIAM Symposium on Discrete
  Algorithms}, pages 2274--2290, 2018.

\bibitem{babaioff2020fair}
M.~Babaioff, T.~Ezra, and U.~Feige.
\newblock Fair and truthful mechanisms for dichotomous valuations.
\newblock In {\em Proceedings of the 35th AAAI Conference on Artificial
  Intelligence}, pages 5119--5126, 2021.

\bibitem{Barman2018}
S.~Barman, S.~K. Krishnamurthy, and R.~Vaish.
\newblock Greedy algorithms for maximizing {Nash} social welfare.
\newblock In {\em Proceedings of the 17th International Conference on
  Autonomous Agents and MultiAgent Systems}, pages 7--13, 2018.

\bibitem{Bei2021}
X.~Bei, Z.~Li, J.~Liu, S.~Liu, and X.~Lu.
\newblock Fair division of mixed divisible and indivisible goods.
\newblock {\em Artificial Intelligence}, 293(103436):1--17, 2021.

\bibitem{Benabbou2021}
N.~Benabbou, M.~Chakraborty, A.~Igarashi, and Y.~Zick.
\newblock Finding fair and efficient allocations for matroid rank valuations.
\newblock {\em ACM Transactions on Economics and Computation},
  9(4):21:1--21:41, 2021.

\bibitem{Bhaskar+2021}
U.~Bhaskar, A.~Sricharan, and R.~Vaish.
\newblock On approximate envy-freeness for indivisible chores and mixed
  resources.
\newblock In {\em Proceedings of Approx}, 2021.

\bibitem{Caragiannis2019}
I.~Caragiannis, D.~Kurokawa, H.~Moulin, A.~D. Procaccia, N.~Shah, and J.~Wang.
\newblock The unreasonable fairness of maximum {Nash} welfare.
\newblock {\em ACM Transactions on Economics and Computation},
  7(3):12:1--12:32, 2019.

\bibitem{Cole2017}
R.~Cole, N.~Devanur, V.~Gkatzelis, K.~Jain, T.~Mai, V.~V. Vazirani, and
  S.~Yazdanbod.
\newblock Convex program duality, {Fisher} markets, and {Nash} social welfare.
\newblock In {\em Proceedings of the 18th ACM Conference on Economics and
  Computation}, pages 459--460, 2017.

\bibitem{Cole2015}
R.~Cole and V.~Gkatzelis.
\newblock Approximating the {Nash} social welfare with indivisible items.
\newblock In {\em Proceedings of the 47th Annual ACM Symposium on Theory of
  Computing}, pages 371--380, 2015.

\bibitem{Cole2018}
R.~Cole and V.~Gkatzelis.
\newblock Approximating the {Nash} social welfare with indivisible items.
\newblock {\em SIAM Journal on Computing}, 47(3):1211--1236, 2018.

\bibitem{Darman2015}
A.~Darmann and J.~Schauer.
\newblock Maximizing {Nash} product social welfare in allocating indivisible
  goods.
\newblock {\em European Journal of Operational Research}, 247(2):548--559,
  2015.

\bibitem{Devanur2008}
N.~R. Devanur, C.~H. Papadimitriou, A.~Saberi, and V.~V. Vazirani.
\newblock Market equilibrium via a primal--dual algorithm for a convex program.
\newblock {\em Journal of the ACM}, 55(5):22:1--22:18, 2008.

\bibitem{FM2022a}
A.~Frank and K.~Murota.
\newblock {\BIB{a}}{D}ecreasing minimization on {M}-convex sets: background and
  structures.
\newblock {\em Mathematical Programming}, 195(1--2):977--1025, 2022.

\bibitem{FM2022b}
A.~Frank and K.~Murota.
\newblock {\BIB{b}}{D}ecreasing minimization on {M}-convex sets: algorithms and
  applications.
\newblock {\em Mathematical Programming}, 195(1--2):1027--1068, 2022.

\bibitem{FMdecmin2}
A.~Frank and K.~Murota.
\newblock Decreasing minimization on base-polyhedra: Relation between discrete
  and continuous cases.
\newblock {\em Japan Journal of Industrial and Applied Mathematics},
  40(1):183--221, 2023.

\bibitem{fujishige1980}
S.~Fujishige.
\newblock Lexicographically optimal base of a polymatroid with respect to a
  weight vector.
\newblock {\em Mathematics of Operations Research}, 5(2):186--196, 1980.

\bibitem{fujishige2005}
S.~Fujishige.
\newblock {\em Submodular Functions and Optimization}.
\newblock Elsevier, 2nd edition, 2005.

\bibitem{fujishige2009}
S.~Fujishige.
\newblock Theory of principal partitions revisited.
\newblock {\em Research Trends in Combinatorial Optimization: Bonn 2008}, pages
  127--162, 2009.

\bibitem{GJ1979}
M.~R. Garey and D.~S. Johnson.
\newblock {\em Computers and Intractability: A Guide to the Theory of
  {NP}-Completeness}.
\newblock Freeman New York, 1979.

\bibitem{Garg2018}
J.~Garg, M.~Hoefer, and K.~Mehlhorn.
\newblock Approximating the {Nash} social welfare with budget-additive
  valuations.
\newblock In {\em Proceedings of the 29th Annual ACM-SIAM Symposium on Discrete
  Algorithms}, pages 2326--2340, 2018.

\bibitem{Goko+2022}
H.~Goko, A.~Igarashi, Y.~Kawase, K.~Makino, H.~Sumita, A.~Tamura, Y.~Yokoi, and
  M.~Yokoo.
\newblock Fair and truthful mechanism with limited subsidy.
\newblock In {\em Proceedings of the 21st International Conference on
  Autonomous Agents and Multiagent Systems}, page 534–542, 2022.

\bibitem{halpern2020fair}
D.~Halpern, A.~D. Procaccia, A.~Psomas, and N.~Shah.
\newblock Fair division with binary valuations: One rule to rule them all.
\newblock In {\em Proceedings of the 16th International Conference on Web and
  Internet Economics}, pages 370--383, 2020.

\bibitem{HLLT2006}
N.~J. Harvey, R.~E. Ladner, L.~Lov{\'a}sz, and T.~Tamir.
\newblock Semi-matchings for bipartite graphs and load balancing.
\newblock {\em Journal of Algorithms}, 59(1):53--78, 2006.

\bibitem{IFF2001}
S.~Iwata, L.~Fleischer, and S.~Fujishige.
\newblock A combinatorial strongly polynomial algorithm for minimizing
  submodular functions.
\newblock {\em Journal of the ACM}, 48(4):761--777, 2001.

\bibitem{Lee2017}
E.~Lee.
\newblock {APX}-hardness of maximizing {Nash} social welfare with indivisible
  items.
\newblock {\em Information Processing Letters}, 122:17--20, 2017.

\bibitem{LiLLT2023}
Z.~Li, S.~Liu, X.~Lu, and B.~Tao.
\newblock Truthful fair mechanisms for allocating mixed divisible and
  indivisible goods.
\newblock In {\em Proceedings of the 32nd International Joint Conference on
  Artificial Intelligence}, 2023.
\newblock Forthcoming.

\bibitem{liu2023mixed}
S.~Liu, X.~Lu, M.~Suzuki, and T.~Walsh.
\newblock Mixed fair division: A survey, 2023.

\bibitem{Lu+2023}
X.~Lu, J.~Peters, H.~Aziz, X.~Bei, and W.~Suksompong.
\newblock Approval-based voting with mixed goods.
\newblock In {\em Proceedings of the 37th AAAI Conference on Artificial
  Intelligence}, 2023.
\newblock Forthcoming.

\bibitem{Maruyama1978}
F.~Maruyama.
\newblock A unified study on problems in information theory via polymatroids
  (in {J}apanese), 1978.
\newblock Graduation Thesis.

\bibitem{MHM2007}
S.~Moriguchi, S.~Hara, and K.~Murota.
\newblock On continuous/discrete hybrid {M}${}^\natural$-convex functions (in
  {J}apanese).
\newblock {\em Transactions of the Institute of Systems, Control and
  Information Engineers}, 20:84--86, 2007.

\bibitem{Murota1998}
K.~Murota.
\newblock {Discrete convex analysis}.
\newblock {\em Mathematical Programming}, 83(1-3):313--371, 1998.

\bibitem{Murota:DCA}
K.~Murota.
\newblock {\em Discrete Convex Analysis}.
\newblock Society for Industrial and Applied Mathematics, 2003.

\bibitem{nagano2007}
K.~Nagano.
\newblock On convex minimization over base polytopes.
\newblock In {\em Proceedings of the 12th International Conference on Integer
  Programming and Combinatorial Optimization}, pages 252--266. Springer, 2007.

\bibitem{textbook}
N.~Nisan, T.~Roughgarden, E.~Tardos, and V.~V. Vazirani.
\newblock {\em Algorithmic Game Theory}.
\newblock Cambridge University Press, 2007.

\bibitem{NS2023}
K.~Nishimura and H.~Sumita.
\newblock Envy-freeness and maximum nash welfare for mixed divisible and
  indivisible goods.
\newblock {\em arXiv:2302.13342}, 2023.

\bibitem{Orlin2010}
J.~B. Orlin.
\newblock Improved algorithms for computing {Fisher's} market clearing prices:
  Computing {Fisher's} market clearing prices.
\newblock In {\em Proceedings of the 42nd ACM Symposium on Theory of
  Computing}, pages 291--300, 2010.

\bibitem{Schrijver2000}
A.~Schrijver.
\newblock A combinatorial algorithm minimizing submodular functions in strongly
  polynomial time.
\newblock {\em Journal of Combinatorial Theory, Series B}, 80(2):346--355,
  2000.

\bibitem{Segal-Halevi2019}
E.~Segal-Halevi and B.~R. Sziklai.
\newblock Monotonicity and competitive equilibrium in cake-cutting.
\newblock {\em Economic Theory}, 68(2):363--401, 2019.

\bibitem{THM2004}
Y.~Takamatsu, S.~Hgara, and K.~Murota.
\newblock Continuous/discrete hybrid convex optimization and its optimality
  criterion (in {J}apanese).
\newblock {\em Transactions of the Institute of Systems, Control and
  Information Engineers}, 17:409--411, 2004.

\bibitem{Varian}
H.~R. Varian.
\newblock Equity, envy and efficiency.
\newblock {\em Journal of Economic Theory}, 9:63--91, 1974.

\bibitem{Vegh2016}
L.~A. V\'{e}gh.
\newblock A strongly polynomial algorithm for a class of minimum-cost flow
  problems with separable convex objectives.
\newblock {\em SIAM Journal on Computing}, 45(5):1729--1761, 2016.

\end{thebibliography}

\appendix

\section{The Nash welfare and a symmetric strictly convex function}\label{sec:MNW}
In this section, we describe the detail of the Nash welfare for binary additive valuations and a symmetric strictly convex function.
Namely, we can find the utility vector of an MNW allocation by minimizing some symmetric strictly convex functions.

Recall that an MNW allocation $\pi$ is the one such that the number of agents with positive utilities is maximized, and subject to that, the Nash welfare $\prod_{i\in N: \pi_i(E)>0} \pi_i(E)$ is maximized.
Maximizing the product of utilities is equivalent to maximizing the sum of the logarithm of utilities.
To minimize the number of agents with zero utilities, we set a symmetric strictly convex function $\Phi(z)=-\sum_{i\in N} \log(z_i+\epsilon)$, where $\epsilon$ is a positive real that is smaller than $\frac{1}{(2|E|n!)^{n+1}}$. 
We show that $\Phi$-minimizer for this function indeed implies an MNW allocation.
To this end, we see a property of an MNW allocation.

Suppose that the allocation of indivisible goods in an MNW allocation is given.
We may remove agents who do not receive indivisible goods and do not desire divisible goods because the utility of such agents cannot be positive.
Then we can find an MNW allocation by assuming that each indivisible good is divisible and desired by only the agent who receives it.
For the setting where only divisible homogeneous goods exist, Devanur et al.~\cite{Devanur2008} proposed an algorithm to find an MNW allocation.
We may assume that every agent has a positive valuation for at least one divisible good, and 
hence there exists an allocation in which every agent has a positive utility.
By applying Lemma 7.1 in \cite{Devanur2008} to when all the agents have binary additive utility functions, there is an MNW allocation in which each agent's utility is a rational value with a denominator at most $n$. 
Moreover, by Theorem~\ref{thm:divisible}, every MNW allocation satisfies the above property, and so does every $\Phi$-fair allocation.

Therefore, in a similar way to the proof of Proposition~\ref{prop:Phi}, we obtain the following lemma.
\begin{lemma}\label{lem:lower}
    In every $\Phi$-fair allocation, each agent's utility is a rational value with a denominator at most $n$, and hence it is zero or a positive rational value at least $1/n$.
\end{lemma}

We focus only on the utility vectors that satisfy the property in this lemma.
Let $x$ be the utility vector of an allocation $\pi$ in which $s \ (\geq 1)$ agents have zero utility, and let $y$ be the utility vector of an allocation in which $t \ (< s)$ agents have zero utility. Suppose that $x$ and $y$ satisfy the property.
Then, we have
\begin{align}
    \Phi(y) - \Phi(x) 
    &=-\sum_{i\in N}\log(y_i+\epsilon)+\sum_{i\in N} \log(x_i+\epsilon)\\
    &\leq (s-t)\log \epsilon + (n-s) \log( |E|/(n-s)+\epsilon) - (n-t)\log(1/n)\\
    &=(s-t)\log (n\epsilon) + (n-s)\log \left(n|E|/(n-s)+n\epsilon\right)\\
    &\le \log (n\epsilon) + n\log \left(n|E|+1\right)\\
    &\le \log (n\left(n|E|+1\right)^{n}\epsilon)
    < \log (\left(2n|E|\right)^{n+1}\epsilon)<0,
\end{align}
where the second inequality holds by $\log(n\epsilon)<0$.
Thus, $\Phi$ is minimized at allocations that maximize the number of agents with positive utilities.
Let $u$ and $v$ be the corresponding utility vectors of two such allocations.
By Lemma~\ref{lem:lower}, $u'=n! \cdot u$ and $v'=n! \cdot v$ are integral vectors.
Let $A=\prod_{i\in N: u'_i>0} (u'_i+n!\epsilon) - \prod_{i\in N: v'_i>0} (v'_i+n!\epsilon)$ and $B=\prod_{i\in N: u'_i>0} u'_i - \prod_{i\in N: v'_i>0} v'_i$.
Then, we have
\begin{align}
    |A-B| 
    &\leq \left|\prod_{i\in N: u'_i>0} (u'_i+n!\epsilon) - \prod_{i\in N: u'_i>0} u'_i \right| + \left|\prod_{i\in N: v'_i>0} (v'_i+n!\epsilon) - \prod_{i\in N: v'_i>0} v'_i \right|\\ 
    &\le \epsilon\cdot 2^n\prod_{i\in N: u'_i>0} u'_i  + \epsilon\cdot 2^n\prod_{i\in N: v'_i>0} v'_i\\
    &\le \epsilon\cdot 2\cdot (2|E|n!)^n
    <  (2|E|n!)^{n+1} \cdot \epsilon < 1.
\end{align}
This implies that if $A>0$, then $B \geq 0$ because $B$ is integral.
In other words, if $\Phi(u) < \Phi(v)$ then $\prod_{i\in N: u_i>0} u_i \geq \prod_{i\in N: v_i>0} v_i$. 

Therefore, any $\Phi$-fair allocation is an MNW allocation.

\section{Hardness of allocation for a specified utility vector}
In this section, we demonstrate NP-hardness of checking existence of an allocation whose utility vector is equal to a given vector.
We remark that, when there are only divisible goods or only indivisible goods, the problem can be solved in polynomial time via the maximum flow problem.

\begin{theorem}\label{thm:util-hard}
It is NP-hard to determine whether a given utility vector $u\in\mathbb{R}^N$ can be achieved by an allocation, even when the indivisible goods are identical.
\end{theorem}
\begin{proof}
We provide a reduction from the 3DM problem.
For a given instance of 3DM $(X,Y,Z;T)$, we construct an instance of fair allocation $(N,M,C,v)$ as follows.
Suppose that $X=\{x_1,\dots,x_n\}$, $Y=\{y_1,\dots,y_n\}$, $Z=\{z_1,\dots,z_n\}$, and $T=\{t_1,\dots,t_m\}\subseteq X\times Y\times Z$.
The set of agents is defined as $N=X\cup Y\cup Z\cup T$.
We create $n$ indivisible goods $M=\{g_1,g_2,\dots,g_n\}$, which will correspond to the hyperedges chosen as a matching.
We set the valuations so that only agents in $T$ desire the indivisible goods, i.e., $v_t(g)=1~(\forall t\in T)$ and $v_s(g)=0~(\forall s\in X\cup Y\cup Z)$ for all $g\in M$.
Additionally, we create $m$ divisible goods $C=\{c_1,c_2,\dots,c_m\}$.
For each $i\in[m]$, the divisible good $c_i$ is desired by only four agents $t_i,x,y,z$, where $t_i=(x,y,z)$.
Moreover, we set the utility vector $u$ as $u_t=1$ for all $t\in T$ and $u_s=1/3$ for all $s\in X\cup Y\cup Z$.
We summarize the valuation of each agent for each good and the given utilities in Table~\ref{tbl:reduction-util}.

\begin{table}[htb]
    \centering
    \caption{Agents' valuations and utilities in the proof of Theorem~\ref{thm:util-hard}.}\label{tbl:reduction-util}
    \begin{tabular}{c|ccc}
      \toprule
      agents               & $g_i~(i\in[n])$ & $c_{i}~(i\in[m])$                   & $u$  \\\midrule
      $s\in X\cup Y\cup Z$ & $0$             & $1$ if $s\in t_i$ and $0$ otherwise & $1/3$\\
      $t_j\in T$           & $1$             & $1$ if $i=j$  and $0$ otherwise     & $1$  \\
      \bottomrule
    \end{tabular}
\end{table}

We demonstrate that the 3DM instance is a yes-instance if and only if there is an allocation such that its utility vector is $u$.

Suppose that the 3DM instance is a yes-instance.
Let $T'$ be a subset of $T$ such that each element in $X\cup Y\cup Z$ appears exactly once.
We construct a desired allocation $\pi^*$ as follows.
For the indivisible goods $M$, we distribute one to each agent in $T'$.
For each $t_i\in T'$ with $t_i=(x,y,z)$, we allocate $1/3$ fraction of $c_i$ to each of agents $x,y,z$.
For each $t_i\in T\setminus T'$, we assign the entire $c_i$ to $t_i$.
Then, it is not difficult to see that $\pi^*(E)$ is the desired vector $u$.

Conversely, suppose that $\pi^*$ is an allocation such that $\pi^*(E)=u$.
Let $T^\dagger\subseteq T$ be the set of agents who receive indivisible goods.
We prove that $T^\dagger$ is a perfect matching for the 3DM instance.
Suppose to the contrary that $T^\dagger$ is not a perfect matching in the 3DM instance.
Then, there exist two distinct agents $t_i=(x,y,z)$ and $t_j=(x',y',z')$ in $T^\dagger$ such that they have an intersection.
Then, we have $k\coloneqq|\{x,y,z,x',y',z'\}|<6$.
As $c_i$ and $c_j$ must be allocated to someone in $\{x,y,z,x',y',z'\}$, we have $\sum_{s\in\{x,y,z,x',y',z'\}}\pi_s^*(E)\ge 2$.
In contrast, we also have $\sum_{s\in\{x,y,z,x',y',z'\}}u_s=k/3<2$.
This contradicts the assumption that $\pi^*(E)=u$.
\end{proof}

\end{document}